 \documentclass[10pt]{article}
\usepackage{geometry}                % See geometry.pdf to learn the layout options. There are lots.
\usepackage[parfill]{parskip}    % Activate to begin paragraphs with an empty line rather than an indent
\usepackage{graphicx}
\usepackage{amssymb}
\usepackage{amsmath}
\usepackage{amsthm}
\usepackage{epstopdf}
\usepackage{verbatim}
\usepackage{mathrsfs}
\usepackage{hyperref}

\usepackage{xcolor}

\usepackage{fancyhdr}
\newtheorem{thm}{Theorem}[section]

\newtheorem{rem}[thm]{Remark}
\newtheorem{lem}[thm]{Lemma}
\newtheorem{prop}[thm]{Proposition}

\newtheorem{defa}[thm]{Definition}
\newtheorem{ass}{Assumption}

\DeclareMathOperator*{\esssup}{ess\;sup}

\newcommand{\ind}{1\!\!1}
\newcommand{\nn}{\nonumber}

\newcommand{\f}{\mathcal{F}}
\newcommand{\A}{\mathcal{A}}
\newcommand{\e}{\mathrm{E}}
\newcommand{\R}{\mathbb{R}}

\newcommand{\mep}{\mathbb{P}}
\newcommand{\meq}{\mathbb{Q}}

\newcommand{\ket}{\mathcal{K}_{t,T}}
\newcommand{\cet}{\mathcal{C}_{t,T}}
\newcommand{\zet}{Z_{t,T}}
\newcommand{\zetq}{Z^\meq_{t,T}}
\newcommand{\deta}{\mathcal{D}^\eta_{t,T}}

\newcommand{\ma}{\mathcal{M}^a}
\newcommand{\q}{\mathcal{Q}}
\newcommand{\z}{\mathcal{Z}}

\newcommand{\dqdpt}{\frac{\mathrm{d}\meq}{\mathrm{d}\mep|_{\f_T}}}
\newcommand{\D}{\mathcal{D}}

\newcommand{\alp}{\begin{equation*}\left\{\begin{array}{lcl}}
\newcommand{\dal}{\end{array}\right.\end{equation*}}
\newcommand{\alpn}{\begin{equation}\left\{\begin{array}{lcl}}
\newcommand{\daln}{\end{array}\right.\end{equation}}
\newcommand{\bq}{\begin{equation*}}
\newcommand{\eq}{\end{equation*}}		
\newcommand{\bqn}{\begin{equation}}
\newcommand{\eqn}{\end{equation}}
\newcommand{\bqq}{\begin{eqnarray*}}
\newcommand{\eqq}{\end{eqnarray*}}
\newcommand{\bqqn}{\begin{eqnarray}}
\newcommand{\eqqn}{\end{eqnarray}}

\DeclareMathOperator*{\essinf}{ess\;inf}

\title{Time--consistent investment under model uncertainty: the robust forward criteria\thanks{
The work presented in this paper is part of our ongoing research into optimal portfolio choices under model uncertainty and we welcome and invite all comments. We would like to thank participants in the 
\emph{New developments in stochastic analysis} workshop, CAS, Beijing, July 2013; IGK workshop \emph{Stochastic and real world models} in Bielefeld, July 2013; $6^{\textrm{th}}$ European Summer School in Financial Mathematics, Vienna, September 2013; \emph{Mathematical Finance} seminar at Columbia University and \emph{Financial/Actuarial Mathematics} seminar at University of Michigan for their comments and suggestions.}
}
\author{Sigrid K\"allblad\thanks{CMAP, Ecole Polytechnique, Paris. Email: \texttt{sigrid.kallblad@cmap.polytechnique.fr}. The work was conducted as part of D.Phil.\ thesis at University of Oxford and was supported by Santander Graduate Scholarship and the Oxford-Man Institute of Quantitative Finance.}
\and Jan Ob\l\'oj\thanks{Mathematical Institute, the Oxford-Man Institute of Quantitative Finance and St John's College, University of Oxford, Oxford, UK. Email: \texttt{Jan.Obloj@maths.ox.ac.uk}. The author gratefully acknowledges support from ERC Starting Grant {\sc RobustFinMath} 335421.}
\and Thaleia Zariphopoulou\thanks{Depts.\ of Mathematics and IROM, The University of Texas at Austin, Austin, USA, 78712; Email: \texttt{zariphop@math.utexas.edu}.  The author would like to thank the Oxford-Man Institute of Quantitative Finance, Oxford, for its hospitality and support. She also acknowledges support from NSF RTG-DMS Grant.}
}
\date{First version: September 2013; This version: \today}

\begin{document}

\maketitle

\label{chap3}

\begin{abstract}
We combine forward investment performance processes and ambiguity averse portfolio selection. We introduce the notion of robust forward criteria which addresses the issues of ambiguity in model specification as well as in preferences and investment horizon specification. It describes the evolution of dynamically--consistent ambiguity averse preferences. 

We first focus on establishing dual characterizations of the robust forward criteria. This is advantageous as the dual problem amounts to a search for an infimum whereas the primal problem features a saddle-point. Our approach is based on ideas developed in Schied \cite{schied} and \v{Z}itkovi\'c \cite{gordan}. We then study in detail non-volatile criteria. In particular, we solve explicitly the example of an investor who starts with a logarithmic utility and applies a quadratic penalty function. The investor  builds a dynamic estimate of the market price of risk $\hat \lambda$ and updates her stochastic utility in accordance with the so-perceived elapsed market opportunities. We show that this leads to a time-consistent optimal investment policy given by a fractional Kelly strategy associated with $\hat \lambda$. The leverage is proportional to the investor's confidence in her estimate $\hat \lambda$.

\end{abstract}

\section{Introduction}

This paper is a contribution to optimal investment as a problem of normative decisions under uncertainty. This topic is central to financial economics and mathematical finance, and the relevant body of research is large and diverse. Within it, the expected utility maximisation (EUM), with its axiomatic foundation going back to von Neumann and Morgenstern \cite{MR0021298} and Savage \cite{Savage:1954wo}, is probably the most widely used and extensively studied framework. 
In continuous time optimal portfolio selection it was first explored in Merton \cite{merton}. The resulting problem considers maximisation of expected utility of terminal wealth:
$$ \max_{\pi} \mathrm{E}_{\mathbb{P}}\left[U(X^\pi_T)\right],$$
where $\mathbb{P}$ is the so-called historical probability measure, $T$ the trading horizon, and $U(\cdot)$ the investor's utility at $T$. 

\textbf{Drawbacks of the classical EUM framework}. Despite the popularity of the above model, there has been a considerable amount of criticism of the model fundamentals $(\mathbb{P},T, U)$, for these inputs might be ambiguous, inflexible and difficult to specify. 
Firstly, an investor faces a significant ambiguity as to which market model to use, specifically, how to determine the probability measure $\mathbb{P}$. This is often referred to as the \emph{Knightian uncertainty}, in reference to the original contribution of Knight \cite{Knight:21}. Introduction of \emph{ambiguity aversion} axiom, motivated by Ellsberg \cite{ellsberg1961risk} paradox, led to generalised \emph{robust} EUM paradigm in Gilboa and Schmeidler \cite{gilboa89}. It built on earlier contributions, including Anscombe and Aumann \cite{anscombe1963definition} and Schmeidler \cite{schmeidler1989subjective}, and has since been followed and extended by a large number of works; we refer the reader to Maccheroni et al.\ \cite{maccheroni06}, Schied \cite{schied} and to F\"ollmer, Schied and Weber \cite{tre} and the references therein for an overview.

Secondly, the investment horizon $T$ might not be fixed and/or a priori known. Such situations arise, for example, in investment problems with rolling horizons or problems in which the horizon needs to be modified
due to inflow of new funds, new market opportunities, or new investment options and obligations. One of the issues related to flexible trading horizons is under which model conditions and preference structure one
could extend the standard investment problem beyond a pre-specified horizon in a time-consistent
manner. This question was recently examined by K\"allblad \cite{kallblad}. The flexibility of investment horizon is also directly related to utilities that are not biased by the horizon choice. The concept of \emph{horizon-unbiased utilities} was introduced by Henderson and Hobson \cite{henderson07}; see also Choulli et al. \cite{choulli2007minimal}.

Thirdly, there are various issues with regards to the elucidation, specification and choice of the utility function. Covering all existing works is beyond the scope herein and we only refer to representative lines of research. Firstly, the concept of utility per se might be quite elusive and one should look for different, more pragmatic criteria to use in order to quantify the risk preferences of the investor. We refer the reader to an old note of F. Black \cite{black68} where the criterion is the choice of the optimal portfolio, see also He and Huang \cite{HeHuang:94} and Cox, Hobson and Ob\l\'oj \cite{CoxHobsonObloj:12}, and to Sharpe \cite{sharpe12} and Monin \cite{monin2012dynamic} where the criterion is a targeted wealth distribution. Another line of research accepts the utility as an appropriate device to rank outcomes but challenges the classical EUM, for empirical evidence shows that investors feel differently with respect to gains and losses. Among others, see, Hershey and Schoemaker \cite{Hershey:1985bn} and Kahneman and Tversky \cite{Kahneman:1979wl}. This prompted further ramifications and led to the development of the area of behavioural finance (see, among others, Barberis \cite{barberis2003survey} and Jin and Zhou \cite{Jin:2008ek}). A third line generalises the concept of utility and moves away from a terminal-horizon deterministic utility, as $U(\cdot)$ above, by allowing state- and path-dependence. One of the best known paradigm are the \emph{recursive utilities}, see, among others, \cite{duffie1992stochastic,el1997backward,skiadas03}.
% which are stochastic processes constructed via a utility-generator. 
They alleviate several drawbacks of their standard counterparts and have been widely used. State-dependent utilities have been also considered before in static frameworks (see, for example, \cite{dreze1961fondements,karni1985decision,karni1983state}). A new family of state-dependent utilities are the so-called \emph{forward investment performance processes}, recently introduced by Musiela and Zariphopoulou \cite{zari10,zarispde}. Their key property is that they are created forward in time. They are stochastic processes $U(\cdot, T)$ which are defined for all time horizons and thus alleviate the horizon inflexibility. More importantly, they are flexible enough to incorporate changing market opportunities, investorÕs views, benchmarked performance, stochastically involved risk appetite and risk aversion volatility, and specification of present utility rather than utility in the (possibly remote) future. We refer the reader to Musiela and Zariphopoulou \cite{zari10,zarispde} for an overview of the topic.

\textbf{Our motivation and objective}. Our work herein was motivated by the above considerations of the triplet of model inputs $(\mathbb{P},T, U)$. We propose a framework that addresses simultaneously the above drawbacks and combines elements of the robust EUM and the forward performance approaches presented above.

Specifically, we consider an investor who invests in a stochastic market in which she does not know the ``true'' model, nor even if such a true model exists. Instead, she describes the market reality through relative weighting of stochastic models with some models being more  likely than the others, some being excluded all together, etc. These views are expressed by a penalty function and are updated dynamically with time. 
The investor's personal evaluation of wealth is expressed through her utility function. We adopt the axiomatic approach to normative decisions which implies that, when considering a given investment horizon, say $T$, the investor aims to maximise the robust expected utility functional, as in Maccheroni et al.\ \cite{maccheroni06} and Schied \cite{schied}. We then generalise this criterion by considering a stochastic extension, which is defined for all investment horizons. 

We combine the classical approaches to Knightian uncertainty and robust utility maximisation with forward investment performance criteria. These criteria evolve forward in time in contradistinction with the existing ones, which are pre-specified up to a certain horizon and are generated backwards using the Dynamic Programming Principle. Such dynamic-consistency, otherwise known as the self-generation property, see \v{Z}itkovi\'c \cite{gordan} and Zariphopoulou and \v{Z}itkovi\'c \cite{zariphopoulou10}), is a natural property linked to optimality both in robust and model-specific EUM. In contrast, it needs to be imposed in the forward investment framework. It is in fact the fundamental element in their very definition. 
 
\textbf{Main contribution}. 
We investigate pairs of utility fields and penalty functions which are dynamically consistent. Such pairs encode stochastic preferences evolving forward in time and taking account of model ambiguity. Accordingly, we call them \emph{robust forward criteria}. 
Our contribution is twofold. First, our theoretical focus is on defining and further characterising the new investment criteria. We consider their duals and establish an appropriate duality result by 
combining ideas developed in Schied \cite{schied} and \v{Z}itkovi\'c \cite{gordan}. As it is the case in existing works, the study of the dual (robust forward) problem offers various advantages. In particular, in the case of robust preferences the dual problem amounts to the search for an infimum whereas the primal problem features a saddle-point. We use the dual formulation to study the question of time-consistency of the optimal strategies. We show that in general, both in our framework as well as in the classical robust EUM, the optimal strategies may fail to be time-consistent. This is caused by possibly arbitrary dynamics of the penalty functions. We show that  time-consistency of the optimal strategies is guaranteed under suitable assumptions of dynamic-consistency of the penalty functions.

Second, we construct a specific example and solve it explicitly. Namely, we consider an investor who starts with a logarithmic utility and applies a quadratic penalty function. Naturally, our solution shows that family of robust forward criteria is non-empty. More importantly, this example offers a theoretical justification and explanation to strategies often followed by large investors in practice. Specifically, the investor aims at building a dynamic estimate of the market price of risk, say $\hat \lambda$, and updates her stochastic utility in accordance with the so-perceived elapsed market opportunities. We show that this leads to a time-consistent optimal investment policy given by a (time-consistent) fractional Kelly strategy associated with $\hat \lambda$. The leverage is a function of investor's confidence in the estimate $\hat \lambda$.

\textbf{Structure of the paper}. The paper is organised as follows. In Section \ref{model3}, the market model is specified and the notion of robust forward criteria is introduced. In Section \ref{secduality}, equivalent dual characterizations of robust forward criteria are established. We also discuss natural examples of penalty functions, including ones associated with risk measures, and link between dynamic-consistency of penalty functions and time-consistency of optimal investment strategies. Then, in Section \ref{secsmoothlog}, working within a Brownian filtration, we study specific classes and examples of robust forward criteria. Our main example is developed in Section \ref{sec:log_ex} where we show how non-volatile logarithmic preferences lead to fractional Kelly strategies. We then discuss a simple example of criteria leading to time inconsistent optimal investment strategies. The remainder of the section is devoted to a, mostly formal, discussion of various classes of criteria. Our aim is to illustrate the flexibility of the notion and the fact that interesting preferences might be identified under additional evolutionary requirements. In particular, non-volatile criteria are linked to a specific PDE which, formally, is discussed in further detail. Finally, we argue that for each robust forward criterion, there exists a specific (standard) forward criterion in the reference market, giving rise to the same optimal behaviour. Most of the proofs are deferred to Section \ref{sec:proofs}.

\section{The market model and the robust forward criterion}\label{model3}

	\subsection{The market model and notation}\label{secmodel}

	The market consists of $d+1$ securities whose prices $(S^0;S)=(S^0_t,S^1_t,...,S^d_t)_{t\in[0,\infty)}$ are modeled as a $(d+1)$-dimensional c\`adl\`ag semi-martingale on a filtered probability space $(\Omega,\f,\mathbb{F},\mep)$, where the filtration $\mathbb{F}=(\f_t)_{t\in[0,\infty)}$ satisfies the usual conditions. We let $S^0\equiv 1$ and assume $S$ to be locally bounded. An $\mathbb{F}$-predictable process $\pi=(\pi_t)_{t\in[0,\infty)}$ is said to be an admissible portfolio if $\pi$ is $S$-integrable on $[0,T]$ for each $T>0$. The associated wealth-process $X^\pi$ is given by
	\bq
		X^\pi_t=\int_0^t\pi_udS_u, \qquad t\ge 0.
	\eq
The set of admissible portfolio processes available to the investor is denoted by $\A$. It is further specified below for the different cases we consider.

	For each $T>0$, $\mathcal{M}^e_T$ denotes the set of equivalent local martingale measures. That is to say, the set of measures $\meq$ on $\f_T$ such that $\meq\sim\mep|_{\f_T}$ and each component of $S$ is a $\meq$--local martingale. Similarly, $\ma_T$ denotes the set of absolutely continuous local martingale measures. The corresponding sets of density processes are denoted respectively by $\z^e_T$ and $\z^a_T$,  
		\bq
			\z^e_T=\left\{Z=\dqdpt: \meq\in\mathcal{M}^e_T\right\}
		\eq
	and similarly for $\z^a_T$. Following \cite{gordan}, we assume that the set $\mathcal{M}^e_T$ is non-empty for each $T>0$. This assumption is referred to as the absence of arbitrage (FLVR) on finite horizons; see Section 2 in \cite{gordan} for further discussion. 
	Note that while
		\bq
			\mathcal{M}^e_{T_1}=\{\meq|_{\f_{T_1}}:\meq\in\mathcal{M}^e_{T_2}\},\quad \textrm{for all $0\le T_1\le T_2$},
		\eq
	there might \emph{not} exist a set $\mathcal{M}^e$ of probability measures equivalent to $\mep$ such that $\mathcal{M}^e_{T}=\{\meq|_{\f_{T}}:\meq\in\mathcal{M}^e\}$, for all $T>0$.
	
	As argued in \cite{gordan}, the condition of NFLVR on finite horizons implies that, for each $\meq\in\mathcal{M}^e_T$, the density process $Z^\meq_t=\e\big[\dqdpt|\f_t\big]$, $t\in[0,T]$, might be extended to a strictly positive martingale $(Z_t)_{t\in[0,\infty)}$ such that $Z_0=1$ and $ZS$ is a local martingale. The set of all such processes $Z$ will be denoted by $\z^e$. In particular, NFLVR on finite horizons holds if and only if $\z^e$ is non-empty. Furthermore, if the condition of strict positivity is replaced by the one of non-negativity, the obtained family is denoted by $\z^a$. For any $\meq\ll\mep$, we use the notation $\zet^\meq:=Z_T^\meq/Z_t^\meq$, with the convention that $\zet^\meq\equiv 1$ on $\{Z_t^\meq=0\}$.

	\subsection{Robust forward performance criteria}\label{secpenalty}

	We introduce the new concept of robust forward criteria. These performance criteria combine two elements: a utility random field $U(\omega,x,t)$, $t\ge 0$, and a family of penalty functions $\gamma_{t,T}(\meq)$, for $0\le t\le T$ and $T\geq 0$. 
	$U(\omega, \cdot ,t)$ models the utility of investor at time $t$ and may depend on the past $(\omega_s)_{s\leq t}$. 
 The investor faces ambiguity about the ``true model" for dynamics of financial assets and forms a view about the relative plausibility of different probability measures. This is reflected in $\gamma_{t,T}(\meq)(\omega)$ which gives the weighting of measure $\meq$ on $\f_T$. Both $U$ and $\gamma$ are combined in making investment decisions. We define both separately before turning to the crucial consistency condition which defines a robust forward criteria couple $(U,\gamma)$.\\
		
\begin{defa}\label{field}
	For a fixed $a\in\{0,\infty\}$, a random field is a mapping $U:\Omega\times(-a,\infty)\times[0,\infty)\to\R$, which is measurable with respect to the product of the optional $\sigma$-algebra on $\Omega\times[0,\infty)$ and $\mathcal{B}\big((-a,\infty)\big)$. A utility random field is a random field which satisfies the following conditions:
		\begin{itemize}
			\item[i)]{For all $t\in[0,\infty)$, the mapping $x\to U(\omega,x,t)$ is $\mep$-a.s.\ a strictly concave and strictly increasing $C^1(\R)$-function which satisfies the Inada conditions
				\bq
					\lim_{x\to-a}\frac{\partial}{\partial x}U(\omega,x,t)=\infty, \qquad
					\lim_{x\to\infty}\frac{\partial}{\partial x}U(\omega,x,t)=0,
				\eq}
			\item[ii)]{For all $x\in(-a,\infty)$, the mapping $t\to U(\omega,x,t)$ is c\`adl\`ag on $[0,\infty)$,}
			\item[iii)]{For each $x\in(-a,\infty)$ and $T\in[0,\infty)$, $U(\cdot,x,T)\in L^1(\f_T)$.\\}
		\end{itemize}
\end{defa}

For a given utility random field, a set of strategies $\A$ is said to be feasible, if for all $\pi\in\A$ and $t>0$, $X_t^\pi\in(-a,\infty)$, a.s.	
	In what follows, we suppress $\omega$ from the notation and simply write $U(x,t)$.
	%Next, we define the notion of admissible penalty functions.
	\\

	\begin{defa}\label{penaltydef}
	For given $t\le T<\infty$, a mapping $\gamma_{t,T}:\Omega\times\{\meq\sim\mep|_{\f_T}\}\to\R_+\cup\{\infty\}$, is called a penalty function if
		\begin{itemize}
			\item[i)]{$\gamma_{t,T}$ is $\f_t$-measurable,}
			\item[ii)]{$\meq\to\gamma_{t,T}(\meq)$ is convex a.s,}
			\item[iii)]{for $\kappa\in L^\infty_+(\f_t)$, $\meq\to\e[\kappa\gamma_{t,T}(\meq)]$ is weakly lower semicontinuous on $\{\meq\sim\mep|_{\f_T}\}$.}
		\end{itemize}
	Moreover, for a given utility random field $U(x,t)$ and feasible set of strategies $\A$, we say that  $(\gamma_{t,T})$, $0\le t\le T<\infty$, is an admissible family of penalty functions if for all $T>0$ and $\pi\in\A$, $\e^\meq[U(X_T^\pi,T)]$ is well defined in $\R\cup\{\infty\}$ for all $\meq\in\q_{t,T}$, $t\le T$, where $\q_{t,T}$ is the set of measures on $\f_T$ given by
		\bqn
			\q_{t,T}:=\left\{\meq\sim\mep|_{\f_T}\textrm{ and $\gamma_{t,T}(\meq)<\infty$ a.s.}\right\}.\label{qtt}
		\eqn
	\end{defa}

	In the above definition, $\q_{t,T}$ is the set of feasible measures considered at time $t$ when investing over $[t,T]$. It may depend on $t$ and $T$ but is non-random. Both larger and smaller sets could be used, e.g.\ the (random) set of measures $\meq$ with $\gamma_{t,T}(\meq)(\omega)<\infty$ or the set of measures $\meq$ with $\e\big[\gamma_{t,T}(\meq)\big]<\infty$. However, for many natural penalty functions, these different choices lead to the same value function, see Section \ref{rmpenalty} below.  Finally, note that we do not impose any regularity or consistency assumptions on $\gamma_{t,T}(\meq)$ in the time variables. These are not necessary for the abstract results in Section \ref{secduality} and will be introduced later when they appear naturally, see Assumption \ref{gammatc}.
	
	We are now ready to introduce the robust forward criteria. As highlighted above, these are couples $(U,\gamma)$ which exhibit a dynamic-consistency akin to dynamic programming principle.\\	

\begin{defa}\label{forward}
	Let $U$ be a utility random field, $\A$ a feasible set of strategies and $\gamma$ an admissible family of penalty functions. Then, the value field associated with $U$, $\A$ and $\gamma$ is a family of mappings \phantom{.} $\{u(\cdot;t,T):0\le t\le T<\infty\}$, with $u(\cdot;t,T):L^\infty(\f_t)\to L^0(\f_t;\R\cup\{\infty\})$ given by
	\bqn
		u(\xi;t,T):=
		\esssup_{\pi\in\A}\essinf_{\meq\in\q_{t,T}}
		\bigg\{\e^\meq\bigg[U\bigg(\xi+\int_t^T\pi_sdS_s,T\bigg)\bigg|\f_t\bigg]
		+\gamma_{t,T}(\meq)\bigg\}, \quad \textrm{for}\; \xi\in L^\infty(\f_t).\label{u3}
	\eqn
	For a given set of feasible strategies, we say that the combination of a utility random field and a family of penalty functions is a \emph{robust forward  criterion} if
	\bqn
		U(\xi,t)=u(\xi;t,T)\quad\textrm{a.s.,}
		\label{eq:dynamic_consistency}
	\eqn
	for all $0\le t\le T<\infty$ and all $\xi\in L^\infty(\f_t)$.
\end{defa}

We note that the above definition is well posed. Indeed, given the assumptions on $U$ and $\gamma$, the conditional expectations in \eqref{u3} are well-defined (extended valued) random variables (see e.g. Prop 18.1.5 in \cite{simonnet96} or p. 113 in \cite{heinz96} for the definition of conditional expectations of quasi-integrable random variables). As all $\meq\in\q_{t,T}$ are equivalent to $\mep$, it also holds for each $\pi\in\A$ that the essential infimum is well-defined (extended-valued) with respect to the reference measure $\mep$.

Optimisation in \eqref{u3} fits within the robust EUM paradigm. Its use to make investment decisions was considered, for a fixed horizon $t\in [0,T]$, in \cite{schied}. It is based on an axiomatic characterisation of risk and ambiguity averse preferences and their numerical representation as concave monetary utility functionals in Maccheroni et al.\ \cite{maccheroni06}, and the robust representation thereof derived in F\"ollmer and Schied \cite{follmer02}. Equation \eqref{eq:dynamic_consistency} provides a direct extension of the notion of self-generating utility fields studied in \cite{gordan} and, consequently, of the notion of forward performance criteria as discussed in the Introduction, see also Section \ref{secsmoothlog}. Accordingly, we sometimes refer to a robust forward criteria $(U,\gamma)$ as being \emph{self-generating} or \emph{dynamically-consistent}. To relate \eqref{eq:dynamic_consistency} to the more classical dynamic programming principle, note that when applied together with the definition \eqref{u3} it yields
	\begin{equation}\label{eq:classicalDC}
\begin{split}
u(\xi,t,T)&=U(\xi,t)=u(\xi,t,r) = \esssup_{\pi\in\A}\essinf_{\meq\in\q_{t,r}}
		\bigg\{\e^\meq\bigg[U\bigg(\xi+\int_t^r\pi_sdS_s,r\bigg)\bigg|\f_t\bigg]
		+\gamma_{t,r}(\meq)\bigg\}\\
		& = \esssup_{\pi\in\A}\essinf_{\meq\in\q_{t,r}}
		\bigg\{\e^\meq\bigg[u\bigg(\xi+\int_t^r\pi_sdS_s,r,T\bigg)\bigg|\f_t\bigg]
		+\gamma_{t,r}(\meq)\bigg\},\quad 0\leq t\leq r\leq T,
\end{split}
\end{equation}
for $\xi\in L^\infty(\f_t)$.

To the best of our knowledge \eqref{u3} corresponds to the most general robust EUM setting which has been previously considered for optimal investment decisions. However we note that this setup has its limitations. For example, the penalty associated to a given measure, $\gamma_{t,T}(\meq)$, is fixed and independent of wealth. This has important implications for time-consistency of optimal investment strategies. We show below in Proposition \ref{prop:time_consistency} that when $(\gamma_{t,T})$ are dynamically consistent, and if we have saddle points $(\pi^{t,T},\meq^{t,T})$ solving \eqref{u3}, then 
%$\meq^{t,T}=\meq^{r,T}_{|\f_t}$
$\meq^{t,r}=\meq^{t,T}|_{\f_r}$, $t\leq r\leq T$, and also the optimal investment strategies are time-consistent. However, in all generality we could have (dynamically consistent) robust forward criteria which lead to time inconsistent optimal strategies. An example is given in Section \ref{sec:ex:non_dc}.
Independence of $\gamma_{t,T}(\meq)$ from investor's wealth is also contrary to the empirical evidence, as discussed in  behavioural finance, see e.g.\ Kahneman and Tversky \cite{Kahneman:1979wl}, which points to the importance of investor's reference point for judging scenarios. In consequence, we believe it might be interesting to study generalisations of the problem in \eqref{u3}. Within the framework of robust EUM, these are possible using quasi-concave utility functionals introduced in Cerreia-Vioglio et al.\ \cite{cerreia}. Their use for (classical) optimal investment problem is being investigated in a parallel paper, see K\"allblad \cite{kallblad2013risk}.

	The set of admissible strategies $\A$ is specified below for the respective cases we consider. Note that the definition of robust forward criteria does not require existence of optimal investment strategies. In that aspect we follow the approach in \cite{gordan} rather than the original definition (cf. \cite{zari10,zarispde}) which required the optimum to be attained. As argued below (cf. Section \ref{secduality}), this flexibility is of particular use for the study of robust forward criteria defined on the entire real line\footnote{For further remarks on the flexibility obtained with this approach, we refer to Remark 3.8 in \cite{gordan}.}. In Section \ref{secsmoothlog} we consider a robust forward criterion of logarithmic type for which the existence of an optimizer is established.

\section{Dual characterization of robust forward criteria}\label{secduality}

Dual methods are well known to be useful for the study of optimal investment problems. For the standard utility maximization problem, they are particularly useful for proving existence of and characterizing the optimal strategy for the primal problem. As we will see below, in our setup there are also clear benefits in passing to the dual domain, even though our focus is on the evolution of the preferences themselves rather than on the optimal strategy. Here the dual problem amounts to a search for an infimum whereas the primal problem features a saddle-point. In consequence, the robust forward criteria are easier to characterize in the dual rather than the primal domain. The aim of this section is to establish such equivalent characterizations. We adopt a convenient set of assumptions with possible extensions discussed in Remarks \ref{dissapoint1} and \ref{dissapoint2} below.

	\subsection{Self-generation in the dual domain} 
	
	We develop the duality theory for utility random fields which are finite on the entire real line. To this end, we set for Definitions \ref{field} and \ref{forward}
		\bq 
		\textrm{$a=\infty$\qquad and \qquad $\A=\A_{bd}$,}
		\eq 
	where $\mathcal{A}_{bd}$ denotes the set of all portfolios producing bounded wealth-processes. Specifically, $\mathcal{A}_{bd}=\mathcal{\bar A}\cap(\mathcal{-\bar A})$, where $\mathcal{\bar A}$ is the set of all admissible portfolio processes for which, for any $T>0$, there exists a constant $c>0$ such that $X^\pi_t\ge -c$, $0\le t\le T$, a.s. The restriction to bounded wealth processes implies that, for many utility fields, the supremum will not be attained. However, this is not really restrictive\footnote{Indeed, the utility field defined on the entire real line does not possess any singularities (cf.\ Assumption \ref{nonsing2} below). The value field defined with respect to a more general (but feasible) set of admissible strategies would therefore coincide with the one defined with respect to bounded strategies. Definition \ref{forward} would still apply, since the notion of robust forward criteria is a consistency requirement placed on the preferences themselves, without a reference to an optimal strategy. In consequence, for utility fields defined on the entire real line, robust forward criteria may be studied and characterized without exactly specifying the domain of optimization.}.
	
	The reason for developing the duality theory for utility random fields finite on the entire real line, is twofold. First, we complement the work of Schied \cite{schied}, since our results are related and therein only utilities defined on the positive half-line are considered. 
	%For contributions to ambiguity averse portfolio-optimization on the entire real line, we refer to \cite{gundel06} and \cite{owari12}. 
	Second, considering utilities finite on the entire real line simplifies certain aspects of the duality theory. This fact is also exploited in, among others, \cite{gundel06}. What usually becomes more complex when allowing for negative wealth, is the definition of an appropriate set of admissible strategies yielding the existence of an optimizer\footnote{Within the present framework where the preferences are not only finite on the entire real line but in addition to that stochastic, the exact specification of a feasible set of admissible, but not necessarily bounded, strategies is highly non-trivial.} (cf. \cite{owari12,schachermayer01}). However, as argued above, for the present purposes it suffices to restrict to the set of bounded wealth processes $\A_{bd}$. In consequence, we may fully benefit from the simplifications of this setup without any further complexity being imposed. While analogous results could be pursued for utilities defined on the half-line, it would imply additional technicalities and we leave it for future research (cf. Remark 3.2 in \cite{gordan}).

	For a given utility random field $U$, the associated dual random field $V:\Omega\times [0,\infty)\times(0,\infty)\to\R$, is given by
	\bqn
		V(y,t)=\sup_{x\in\R}\big(U(x,t)-xy\big)\qquad \textrm{for $t\ge 0$, $y\ge 0$}.\label{vd}
	\eqn
	The dual value field and the notion of self-generation in the dual domain are then naturally defined as follows.\\

\begin{defa}
	For $y>0$ and $0\le t<T<\infty$ the dual value field $v(\cdot;t,T):L^0_+(\f_t)\to L^0(\f_t;\R\cup\{\infty\})$, is given by
	\bqn
		v(\eta;t,T):=
		\essinf_{\meq\in\q_{t,T}}
		\essinf_{Z\in\z^a_T}
		\Big\{\e^\meq\Big[V\left(\eta\zet/\zet^\meq,T\right)\Big|\f_t\Big]
		+\gamma_{t,T}(\meq)\Big\}.\label{lv}
	\eqn
	The combination of a dual random field $V$ and a family of penalty functions $\gamma$ is said to be \emph{self-generating} or \emph{dynamically consistent} if
	\bq
		V(\eta,t)=v(\eta;t,T),\;\;\textrm{a.s.,}
	\eq
	for all $0\le t\le T<\infty$ and all $\eta\in L^0_+(\f_t)$.
	\end{defa}

	\subsection{Equivalence between primal and dual robust self-generation}\label{sec:equivalence}

	We first introduce the following technical assumption:\\
	
	\begin{ass}\label{nonsing2}For each $T>0$ and $0\le t\le T$, the set $\q_{t,T}$ is convex and weakly compact and the set $\{ZU^-(x,T):Z\in\q_{t,T}\}$ is UI, for all $x\in\R$. Furthermore, if $\kappa\in L^\infty_+(\f_t)$ and $\meq\in\q_{t,T}$ are such that $\kappa\zet^\meq U(x,T)\in L^1$, for all $x\in\R$, then
		\bqn
			\tilde U(x,T):=\ind_{\{\kappa=0\}}U(x,T)+\ind_{\{\kappa>0\}}\zet^\meq U(x,T),\quad x\in\R,\label{auxutility}
		\eqn
	%	\bq Y_T(x):=\left\{\begin{array}{lll}
	%		U(T,x) & \{\kappa=0\} \\
	%		\zet^\meq U(T,x) &\{\kappa>0\}
	%		\end{array}\right. x\in\R,
	%	\eq
	satisfies the non-singularity Assumption 3.3 in \cite{gordan}.
	\end{ass}
	
	The above implies that $U(x,t)$ itself satisfies the non-singularity assumption. For further discussion of this concept, we refer to Remark 3.4 in \cite{gordan}. Given that the set $\q_{t,T}$ is weakly compact, a sufficient condition for Assumption \ref{nonsing2} to hold, is that $U(x,t)$ is $(x,\omega)$-uniformly bounded from below by a deterministic utility function. Then, it also trivially holds that any family of penalty functions is admissible. Note also that, due to convexity, the weak compactness of $\q_{t,T}$ is equivalent to closedness in $L^0$ (cf. Lemma 3.2 in \cite{wu}).

	Next, we present the first main result, which yields the conjugacy relations between the functions $u(x;t,T)$ and $v(y;t,T)$. We stress that even for $t=0$, Theorem \ref{main} differs from Theorem 2.4 in \cite{schied} in that the utility function is defined on the entire real line and is also allowed to be stochastic. Moreover, we do not impose any finiteness assumptions on the involved value fields.
\\

	\begin{thm}\label{main}	
	Let $U(x,t)$, $t\ge 0$, be a utility random field, $\gamma_{t,T}$ an admissible family of penalty functions and $V(y,t)$ the associated dual random field. Assume that Assumption \ref{nonsing2} holds. 
	
	Then, for all $\xi\in L^\infty(\f_t)$, $\eta\in L^0_+(\f_t)$ and $0\le t\le T<\infty$, the following assertions hold,
		\bqn
			u(\xi;t,T)=\essinf_{\eta\in L^0_+(\f_t)}\big(v(\eta;t,T)+\xi\eta\big)\quad\textrm{a.s.} \label{a5}
		\eqn 
		and
		\bqn
			v(\eta;t,T)=\esssup_{\xi\in L^\infty(\f_t)}\big(u(\xi;t,T)-\xi\eta\big)
			\quad\textrm{a.s.}\label{a6}
		\eqn
		%\bqqn
		%	u(\xi;t,T)&=&\essinf_{\eta\in L^0_+(\f_t)}\big(v(\eta;t,T)+\xi\eta\big)\quad\textrm{a.s.;} \label{a5}\\
		%	v(\eta;t,T)&=&\esssup_{\xi\in L^\infty(\f_t)}\big(u(\xi;t,T)-\xi\eta\big)
		%	\quad\textrm{a.s.}\label{a6}
		%\eqqn
	In consequence, the combination of a utility random field $U(x,t)$ and a family of penalty functions $\gamma_{t,T}$ is self-generating, if and only if, the combination of the dual random field $V(y,t)$ and $\gamma_{t,T}$ is self-generating.
	\end{thm}

	The proof of Theorem \ref{main} is given in Section \ref{sec:proof_TC} and is based on combining ideas introduced in \cite{schied} and \cite{gordan}, respectively. In the former paper, duality results for the robust utility maximization problem with variational preferences were established. In the latter, in a setting similar to ours, conditional conjugacy relations were established for the non-robust case. Specifically, we reduce the conditional case to an $\f_0$-measurable conjugacy relation by taking expectations. For the latter, the relevant assertions are proven using arguments similar to the ones in \cite{schied}. However, while \cite{schied} relies on the duality results in \cite{kramkov}, we here make use of the theory established in \cite{gordan}.

	As holds for the case of a fixed measure (cf. \cite{gordan}), the dual problem admits a solution even though the primal problem may not (due to the restriction to bounded strategies). The fact that the optimizer's second component is in $\mathcal{M}^a_T$ (as opposed to a larger set of finitely additive measures) is a consequence of the utility function being finite on the entire real line (see \cite{gordan} and also \cite{bellini02,schachermayer01}).\\	
	
	\begin{prop} \label{prop:existence}
		Let $V(y,t)$ be a dual random field such that Assumption \ref{nonsing2} holds for the associated primal field. Then, for each $\eta\in L^0_+$ and $t\le T<\infty$, there exist $\meq\in\q_{t,T}$ and $Z\in\z^a_T$ for which the infimum in the dual value function $v(\eta;t,T)$ is attained (cf. \eqref{lv}). 
	\end{prop}

	We remind the reader that the above results use that the set of measures $\q_{t,T}$, defined in \eqref{qtt}, is assumed to be weakly compact. 	
	We end this section with some remarks on possible further extensions in the definition and assumptions imposed on $\q_{t,T}$.\\

	\begin{rem}\label{dissapoint1}
		Theorem \ref{main} can be proven under the assumption that $\q^a_{t,T}$ is weakly compact, where $\q_{t,T}^a$ is the set of absolutely continuous measures for which the penalty is finite a.s. For example, this holds for all penalty functions associated with coherent risk measures continuous from below (see Section \ref{rmpenalty}). The result then holds with the set $\q_{t,T}$ replaced by $\q^a_{t,T}$ in the definition of $u(\cdot;t,T)$ but with the dual field still defined as above with respect to the equivalent measures. In order to use $\q^a_{t,T}$ in the definition of $v(\cdot;t,T)$, one would need to extend the definition of $Z^\meq V(\eta/Z^\meq)$ to the null-sets of $\meq$ in a suitable way (preserving lower semicontinuity). For the case of utility functions defined on $\R_+$, this was done in \cite{schied}. The present case requires a more careful treatment which is the focus of future research. Extending the definition of the dual problem to a set of absolutely continuous measures would also enable proving Proposition \ref{prop:existence} using the weak compactness of the level sets (cf. Remark \ref{dissapoint2}) rather than of $\q^a_{t,T}$. 
	\end{rem}\smallskip

	\begin{rem}\label{dissapoint2}	
	In \cite{schied}, for the case of positive wealth processes and a fixed time horizon, similar conjugacy relations to \eqref{a5} and \eqref{a6} were established without the compactness assumption on $\q_{t,T}$. The proof exploited instead weak compactness of the level-sets $\q(c):=\{\meq\ll\mep:\gamma_{0,T}(\meq)\le c\}$. Specifically, since $U(\varepsilon+X_T)$, $X_T\ge 0$, is uniformly bounded from below for that case, the infimum in
			\bq
				\sup_{\pi}\inf_{\meq\ll\mep}\left\{\e^{\meq}\left[U\left(\varepsilon+X_T^\pi\right)\right]
				+\gamma_{0,T}(\meq)\right\},
			\eq
		can be replaced by the infimum over some (weakly compact) level set $\q(c)$, $c>0$. After application of a minmax theorem, the result is then obtained by letting $\varepsilon$ go to zero. Since we consider $U:\R\to\R$, the arguments become more involved. Indeed, even for $t=0$, $U(x,T)$ deterministic and $\pi_n$ an optimizing sequence, it is not clear whether $\e[U(X_T^{\pi_n},T)]$ is bounded from below. To address such issues, besides extending the setting from equivalent measures and define the dual problem for absolutely continuous measures, one might have to adopt the more elaborate setup considered in \cite{schachermayer01} where the existence of an optimizer for utility functions defined on the entire real line is proven by defining a sequence of utility functions $U^n$, for each of which the problem is reduced to one defined on the half-line. The result is then obtained by a limiting procedure. We leave these problems for future research.
				\end{rem}

\subsection{Dynamic-consistency of penalty functions and time-consistency of the optimal investment strategies}\label{sec:time_consistent}

	The definition of robust forward criteria requires the \emph{combined} criterion consisting of $U(x,t)$ and $\gamma(\cdot)$ to be dynamically consistent (cf. Definition \ref{forward}). In this section we further investigate this assumption and relate it to the dynamic consistency of the penalty functions and the optimal investment strategies. The proofs of the results in this section are reported in Section \ref{sec:proof_TC}.
	
	We introduce the following class of dynamically consistent penalty functions:\\
		
	\begin{ass}\label{gammatc}For any $T>0$ and $\meq\sim\mep$ on $\mathcal{F}_T$, the family of penalty functions $(\gamma_{t,T})$ is c\`adl\`ag in $t\leq T$, $\gamma_{t,t}\equiv 0$ and
		\bqn
			\gamma_{s,T}(\meq)=\gamma_{s,t}(\meq_{|\f_t})
			+\e^\meq\left[\left.\gamma_{t,T}(\meq)\right|\f_s\right],\quad s\leq t\leq T. \label{kol}
		\eqn
		Moreover, $\tilde\q_{s,T}=\q_{s,T}$, where
		\bqn
			\tilde\q_{s,T}:=\big\{\meq\sim\mep|_{\f_T}: Z^\meq_T=Z^{\meq_0}_tZ^{\meq_1}_{t,T}, \meq_0\in\q_{s,t}, \meq_1\in\q_{t,T}, s\le t\le T\big\}. \label{stable}
		\eqn
	\end{ass}
		
		%Indeed, while it is a sufficient condition that $\gamma$ itself is time-consistent, it might not be a necessary one.
		
	We note that the above property of stability under pasting \eqref{stable} is not implied by \eqref{kol}. In order to render the analysis tractable, we work under this stronger assumption. For remarks on the relation of the above properties to penalty functions associated with risk measures, see Section \ref{rmpenalty} below.

	%\subsubsection{Time-consistency of the value field}
	
	The additional structure resulting from Assumption \ref{gammatc} allows us to consider the question of whether, for $T>0$ fixed, the value field $u(x,t;T)$ associated with a general utility field, is itself self-generating for $t\le T$. That is to say, whether the dynamic programming principle holds (cf. \eqref{eq:classicalDC}). We verify now that under suitable assumptions on the penalty function, this is the case. The proof proceeds by first establishing appropriate consistency in the dual domain and then applying Theorem \ref{main}. \\

	\begin{prop}\label{dcp}
	Let $U(x,t)$ be a utility random field and $\gamma_{t,T}$ an admissible family of penalty functions. Suppose Assumptions \ref{nonsing2} and \ref{gammatc} hold. Then, for each $T>0$, the primal value field $u(\cdot;t,T)$ is self-generating i.e.
	\bqn
		u(x;s,T)=
		\esssup_{\pi\in\mathcal{A}_{bd}}\essinf_{\q_{s,t}}
		\bigg\{\e^\meq\bigg[u\Big(x+\int_s^t\pi_udS_u;t,T\Big)\bigg|\f_s\bigg]
		+\gamma_{s,t}(\meq)\bigg\},\quad 0\leq s\leq t\leq T.\label{tcvalue}
	\eqn	
	\end{prop}
 
		For the case of standard (non--robust) utility maximization and deterministic utility functions it is well-known that the value process satisfies the DPP; also referred to as the martingale optimality principle, see \cite{karoui}. 
		In consequence, standard forward criteria may be seen as a generalization, to all times $t\ge 0$, of value functions associated with stochastic utility functions. 
		Proposition \ref{dcp} shows that a similar consistency property holds for certain ambiguity averse criteria; this has also been used to address ambiguity averse problems by stochastic control arguments in, among others, \cite{hernandez2,hernandez,muller05}. This further justifies our definition of robust forward criteria. 
		
		We recall that the value field associated with a general penalty function may not be dynamically consistent (see \cite{schied} for counter-examples). Hence, while standard forward criteria might be viewed as direct extensions of value functions associated with stochastic utility functions, Definition \ref{forward} enforces a additional structure by imposing the dynamic consistency requirement \eqref{eq:dynamic_consistency} on the couple $(U,\gamma)$. Note that, in general, this is weaker than the assumption of dynamic consistency of $\gamma$. Indeed, in Section \ref{sec:ex:non_dc} below, we construct an example of a dynamically consistent pair $(U,\gamma)$ where the penalty function itself is not. The robust forward criteria may then lead to time \emph{inconsistent} optimal investment strategies. In contrast, when the penalty functions are consistent, we recover the time-consistency of the optimisers.\\

	\begin{prop}\label{prop:time_consistency}
		Let $U(x,t)$ and $\gamma_{t,T}$ be a robust forward criterion such that Assumptions \ref{nonsing2} and \ref{gammatc} hold. Assume further that for each $0\leq t<T<\infty$ and $\xi\in L^\infty(\f_t)$ there is a saddle point $(\pi^{t,T}(\xi),\meq^{t,T}(\xi))$ for which $u(\xi,t;T)$ is attained (cf. \eqref{u3}). Then, the saddle point may be taken to be time consistent in that
		%		re is a process $\bar\pi_t$, $t\ge 0$, such that for all $t<T<\infty$, $u(x,t;T)$ is attained for $\bar\pi^{t,T}:=\bar\pi_s\ind_{s\in[t,T)}$. Further, this optimal strategy is time-consistent in that
		$\meq^{t,T}(\xi)=\meq^{t,\bar T}(\xi)|_{\f_T}$, and
		\bq
			\pi^{t,T}_u(\xi)=\pi^{t,\bar T}_u(\xi)
			\quad\textrm{and}\quad 
			\pi^{t,T}_u(\xi)=\pi^{u,T}_u\left(\xi+\int_t^u \pi_s dS_s\right),
			\quad 0\leq t\le u\le T\le\bar T.
		\eq 
		Further, for $x>0$, there exists a process $\bar\pi_t$, $t\ge 0$, and a positive martingale $Y_t$, $t\ge 0$, such that, for all $0\le t<T<\infty$, $u(x+\int_0^t\bar\pi_sdS_s;t,T)$ is attained for $\pi^{t,T}=\bar\pi_s\ind_{s\in[t,T)}$ and $\bar\meq$, with $\frac{\mathrm{d}\bar\meq}{\mathrm{d}\mep}=Y_T$. 
	\end{prop}

The above result, combined with example in Section \ref{sec:ex:non_dc} shows that the dynamic consistency of penalty functions \eqref{kol} is a necessary and sufficient condition for time-consistency of optimal investment strategies. Further, it is clear from the example that this applies both to the robust forward criteria studied here as well as the classical robust expected utility maximisation on a fixed horizon. This leads to interesting open questions. First, of the economic justification for \eqref{kol} which remains unclear, see Remark 3.5 in \cite{schied}. Second, of generalisations of the optimisation problem in \eqref{u3} which would preserve time-consistency of optimal strategies while \eqref{kol} is violated. The expression in \eqref{u3} arises from the representation of concave utility functionals in \cite{follmer02,fritelli02} and its generalisations correspond to quasi-concave functional explored in \cite{cerreiab,drapeau10}. Their use in the context of utility maximisation is investigated in a parallel work of K\"allblad \cite{kallblad2013risk}. The implications for time-consistency and its link to \eqref{kol} remain however open.

	%\subsubsection{Submartingale property of the dual field} 	
	
	Finally, we show that the dynamic consistency property of penalty functions leads to a characterization of robust forward criteria in terms of a certain ``weighted submartingale" property of the dual field. This will be used to derive an equation allowing us to investigate particular classes of, and to find examples of, the robust forward criteria $(U,\gamma)$.\\

	\begin{prop}\label{characp}
		Let $U(x,t)$ be a utility random field and $\gamma_{t,T}$ an admissible family of penalty functions such that Assumption \ref{nonsing2} holds. In addition, assume either that Assumption \ref{gammatc} holds, or that \eqref{kol} holds and, for all $T>0$, $U(x,T)\in L^1(\f_T,\meq)$ for all $\meq\in\tilde\q_{0,T}$.
		Let $V(y,t)$ the dual field given in \eqref{vd}.
		Then, the following two statements are equivalent: 
			\begin{itemize}
				\item[i)]{$U(x,t)$ and $\gamma_{t,T}$ constitute a robust forward criterion;}
				\item[ii)]{For each $y>0$ and all $t\le T<\infty$, it holds for all $\meq\in\q_{t,T}$ and $Z\in\z^a_T$ that
					\bqn
						V(yZ_t/Z_t^\meq,t)\le \e^\meq\left[\left.V(yZ_T/Z_T^\meq,T)\right|\f_t\right]+\gamma_{t,T}(\meq).\label{eq:villkor}
					\eqn 
			Further, there is $Z\in\z^a$ and a positive martingale $Y_t$, $t\ge 0$, such that, for all $t\le T<\infty$, $\meq_T\in\q_{0,T}$, with $\frac{\mathrm{d}\meq_T}{\mathrm{d}\mep}=Y_T$, and \eqref{eq:villkor} holds as equality for $Z_T$ and $\meq_T$.
			}
			\end{itemize} 	
	\end{prop}

	\subsection{Penalty functions associated with risk measures}\label{rmpenalty}
	
	Recall that preferences specification akin to \eqref{u3} is motivated by results in economics. 		The axiomatic approach
	to ambiguity averse choices under uncertainty led to numerical representation in terms of concave utility functionals, with the penalty function appearing naturally from the robust representation of convex risk measures; see \cite{gilboa89,maccheroni06} and \cite{tre} for an overview. 
	We summarize now some facts about such penalty functions and relate them to our assumptions. To this end, let $\rho_{t,T}$ be a conditional convex risk measure and $\gamma_{t,T}$ its associated minimal penalty function (which we assume to be bounded from below), given by
	\bqn
		\gamma_{t,T}(\meq):=\esssup_{X\in L^\infty(\f_T)}\big(\e^\meq[-X|\f_t]-\rho_{t,T}(X)\big),\label{penaltynew}
	\eqn
	for $\meq\ll\mep|_{\f_T}$. Then, it holds that
	\bqn
		\rho_{t,T}(X)=\esssup_{\substack{\meq\ll\mep|_{\f_T}:\\ \meq|_{\f_t}=\mep|_{\f_t}}}
		\big(\e^\meq[-X|\f_t]-\gamma_{t,T}(\meq)\big),\label{penaltyfcn}
	\eqn
	for $X\in L^\infty(\f_T)$ (see, for example, \cite{bionnadal04,detlefsen05}). Within the context of ambiguity averse portfolio optimization, it is common to restrict to risk measures $\rho_{t,T}$ which are continuous from below, i.e. for $Y^n\in L^\infty$ such that $Y_n\nearrow Y$ a.s. with $Y\in L^\infty$, $\rho_{t,T}(Y_n)\to\rho_{t,T}(Y)$ a.s., and, moreover, ``sensitive``, in that $\mep\left(\e[\rho_{t,T}(-\varepsilon A)]>0\right)>0$, for all $\varepsilon>0$ and $A\in\f_T$ such that $\mep(A)>0$. These properties render, respectively, the associated level sets $\{\meq\ll\mep:\gamma_{0,T}(\meq)<c\}$, $c>0$, weakly compact and $\q_e$ non-empty (cf. Lemma 4.1 in \cite{schied} and Remark \ref{dissapoint2} above).

	We note that $X\in L^\infty$ in \eqref{penaltynew} and \eqref{penaltyfcn}, while $U(X^\pi_T,T)$ only lies in $L^1$. This, however, is not an issue and, in particular, we do not have to restrict $\gamma$ in line with extensions of the risk measure theory to $L^p$-spaces, $p=[1,\infty)$, (see \cite{filipovic07,kaina09} and, for the conditional case, \cite{acciaio11,filipovic12}). Indeed, in analogy with \cite{schied}, it suffices to impose (weaker) joint integrability conditions on $U(x,t)$ and $\gamma$ to ensure that the value function $u(\cdot;t,T)$ is well-defined (cf.\ Definition \ref{penaltydef}).
	
	A penalty function $\gamma_{t,T}$ in \eqref{penaltynew} associated with a risk measure satisfies properties i) - iii) of Definition \ref{penaltydef}. However, in general, it will not satisfy the weak compactness assumptions used above (cf. Assumption \ref{nonsing2}). To illustrate this, note that for this type of penalty functions, it is natural to restrict the set $\q_{t,T}$ in \eqref{qtt} to its subset (cf. e.g. Theorem 1.4 in \cite{acciaio11b}):
	\bqn
		\Big\{\meq\sim\mep|_{\f_T}:
		\e^\meq\big[\gamma_{t,T}(\meq)\big]<\infty\Big\}.\label{skit}	
	\eqn
 	For a general convex risk measure, this set is \emph{not} weakly compact. However, as we consider risk measures which are continuous from below, the associated level sets are. In particular, for a \emph{coherent} risk measure, which corresponds to $\gamma\in\{0,\infty\}$, it follows that
	\bq
		\q^a_{t,T}:=\big\{\meq\ll\mep|_{\f_T}:\meq=\mep\textrm{ on }\f_t\textrm{ and }\gamma_{t,T}(\meq)=0,a.s.\big\}
	\eq
	is weakly compact. If further $\q^a_{t,T}\subseteq\{\meq\sim\mep|_{\f_T}\}$, then the set in \eqref{skit} is also weakly compact. An example of such a risk measure is considered in \cite{hernandez2} (cf. also Theorem 3.16 in \cite{kloppel05}). Naturally, Assumption \ref{nonsing2} allows for much more flexibility.

	For convex risk measures, time-consistency is characterized by property \eqref{kol}. Indeed, \eqref{kol} is equivalent (cf. e.g. Theorem 4.5 in \cite{irina}) to $\rho$, given in \eqref{penaltyfcn}, satisfying, for $0\le s\le t\le T$, 
		\bqn
			\rho_{s,T}(X)=\rho_{s,t}(-\rho_{t,T}(X)).\label{sgphi}
		\eqn
	One would expect this property, combined with Assumption \ref{nonsing2}, to be sufficient for Lemmas \ref{dc} and \ref{charac} to hold. Indeed, assume that $U(x,T)\in L^\infty$, for $x\in\R$. For a fixed strategy $\bar\pi\in\mathcal{A}_{bd}$, the relation in \eqref{tcvalue} then reduces to
	\bqn	\label{eq:extend_tc}
		\phi_{s,T}\big(U\big(X_T^{\bar\pi},T\big)\big)=\phi_{s,t}\big(\phi_{t,T}\big(U\big(X_T^{\bar\pi},T\big)\big),
	\eqn
	where $\phi(X)=-\rho(X)$. Note that \eqref{eq:extend_tc} holds true due to \eqref{sgphi}. Time-consistency of the value function has also been verified for the choice of specific models and utility functions (see, among others, \cite{hernandez2}). We leave  proving our results under this assumption for future research and restrict ourselves to the stronger Assumption \ref{gammatc}. Note that any time-consistent coherent risk measure admits the pasting property \eqref{stable} (cf. Corollary 1.26 in \cite{acciaio11b}). In fact, in our case when all measures in $\q_{t,T}$ are equivalent to the reference measure, even more explicit results hold for these risk measures (for results on the relation between stable sets and time-consistent coherent risk measures, we refer to \cite{delbaen06,irina,kloppel05}).

	\section{On structure, specific classes and examples of robust forward criteria}\label{secsmoothlog}

	Within a Brownian filtration, we consider a logarithmic robust forward criterion with a quadratic penalty structure (cf. Proposition \ref{buk}). The example is of particular interest as it gives theoretical justification to fractional Kelly strategies often used in practice by large investment funds. More precisely, the investor estimates (dynamically) the market growth (Kelly) strategy $\hat X$ and invests a (dynamically adjusted) fraction of her wealth in $\hat X$. The leverage, in our framework, has the interpretation of investor's confidence in his estimate of $\hat X$. 
	
	%In this section we study the effects of additional dynamic requirements on the robust forward criteria.
	%An explicit example which results from this discussion is studied in detail in Section \ref{secsmoothlog}.
	The example belongs to a certain class of so-called non--volatile robust forward criteria. We elaborate further on this in Section \ref{logexex}. Specifically, we provide a formal discussion illustrating the structure of forward criteria and the fact that additional assumptions are needed in order to pin down a unique criterion from a given initial condition and penalty structure. 
	Specific attention is paid to the non--volatile criteria, which are characterized by a specific evolutionary property and linked to a certain PDE (cf. equation \eqref{hjb} below).
	
	Despite its specific form, the example in Section \ref{sec:log_ex} illustrates yet a crucial fact about robust forward criteria. Namely, that for each robust forward criterion, there exists a (standard) forward criterion in the fixed reference market, giving rise to the same optimal behaviour. This is further discussed in Section \ref{sec:equiv}.

	\subsection{Non-volatile criteria yielding fractional Kelly strategies}\label{sec:log_ex}
	
	We first specify the Brownian setup considered throughout this section. At this point, we stress that in reality the investor does not have access to the ``true model", which is an abstract concept. Instead, the investor decides on a reference model $\hat\mep$. 
	In the example below this will be a dynamically updated estimate for the most likely description of reality. It is therefore natural to expect $\gamma_{t,T}(\cdot)(\omega)$ to have a global minimum at $\hat\mep_{|\f_T}$. Further, in the example considered below, we will also see that the randomness of $U(\cdot, x,t)$ is expressed through the realisation of $\hat\mep_{|\f_t}$, $t\geq 0$. 
	
	For simplicity, let $d=1$ in that the market only consists of one risky asset. Recall that $S^0_t\equiv 1$. We consider a filtration generated by a two-dimensional $\hat\mep$-Brownian motion $W_t=(\hat W^1_t,\hat W^2_t)$, $t\ge 0$, and assume that $S^1_t$ solves
		\bqn
			dS^1_t=S^1_t\left(\hat\lambda_tdt+\sigma_td\hat W^1_t\right),\label{model}
		\eqn
	for some $\mathbb{F}$-progressively measurable processes $\sigma_t$, $\sigma_t\neq 0$ a.s., and $\hat\lambda_t$, $t\ge 0$. The latter is referred to as the investor's estimated market price of risk. Further, in this section, we let $(\pi_t)_{t\ge 0}$, denote the fraction of wealth invested in the risky asset. The associated wealth process then follows the dynamics
		\bq	
			dX_t^\pi=\pi_tX_t^\pi\sigma_t\big(\hat\lambda_tdt+d\hat W^1_t\Big),\quad X_0=x.
		\eq
	The set of admissible strategies is defined as follows:
		\bq
			\mathcal{A}:=\Big\{\pi: \textrm{$(\pi_t)$ adapted, $(X_t^\pi)$ well-defined and $X_t^\pi>0$ a.s. for all $t>0$}\Big\},		
		\eq
		and we also write $\mathcal{A}^x$ when we want to stress the initial wealth $X_0=x$. Finally, we denote by $\mathcal{A}_t^x$ the analogue set of strategies on $[t,\infty)$ starting from $X_t^\pi=x$.
		
		Given the Brownian filtration, any measure $\meq\sim\hat \mep$ on $\f_T$ admits a process $\eta_t=(\eta^1_t,\eta^2_t)\in\mathcal{P}\times\mathcal{P}$, $t\le T$, such that $\frac{\mathrm{d}\meq}{\mathrm{d}\hat \mep}\big|_{\f_T}=D^\eta_T$, where the process
		\bqn
			D_t^\eta:=\mathcal{E}\left(\int\eta^1_sd\hat W^1_s
			+\int\eta^2_sd\hat W^2_s\right)_t,\label{d}
		\eqn
	is a martingale on $[0,T]$. We write $\meq=\meq^\eta$ and, for the present example, assign it a penalty given by
		\bqn \label{eq:penalty_ex_noncompact}
		\gamma_{t,T}(\meq^\eta):=\left\{
			\begin{array}{ll}
				\e^{\meq^\eta}\bigg[\int_t^T\frac{\delta_u}{2}\left|\eta_u\right|^2du\bigg|\f_t\bigg] & \textrm{ if } \e^{\meq^\eta}\Big[\int_t^T\hat\lambda_s^2ds\Big]<\infty\\
				+\infty & \textrm{ otherwise,}
			\end{array}
			\right.
		\eqn
		for some adapted, non--negative process $(\delta_t)$ which controls the strength of the penalisation (cf. also \eqref{penalty} below). The investor is aware that $\hat\mep$ may be an inaccurate estimate of the market and $(\delta_t)$ quantifies her trust in $\hat \mep$. Note that $\gamma_{t,T}$ may fail to satisfy Assumptions \ref{nonsing2} and \ref{gammatc}. In particular, $\q_{t,T}$ in \eqref{qtt} may not be weakly compact. This is not a problem since, for this example, we present a direct proof. Finally, we assume that there exists $\kappa>1/2$ such that $\hat\e\big[\exp\big(\kappa\int_0^T\hat\lambda_s^{2}ds\big)\big]<\infty$ for all $T>0$. This is a convenient integrability assumption which can be interpreted as $\hat \mep$ being \emph{reasonable}. Note that it implies in particular, by Novikov's condition, that $(Z^{\nu}_t)$ in \eqref{z} with $\nu\equiv 0$ is a $\hat\mep$-martingale.\\

	\begin{prop}\label{buk}
		Given the investor's choice of $(\hat \lambda_t)$ and $(\delta_t)$ as above, let  	
			\bqn
				\bar\eta_t:=\big[-\hat\lambda_t/(1+\delta_t),0\big]
				\quad\textrm{and}\quad
				\bar\pi_t:=\frac{\delta_t}{1+\delta_t}\frac{\hat\lambda_t}{\sigma_t},\label{etastar}
			\eqn
		and
			\bqn
				U(x,t):=\ln x-\frac{1}{2}\int_0^t\frac{\delta_s}{1+\delta_s}\hat\lambda_s^2ds,\quad t\ge 0, x\in\R_+. \label{mainex}
			\eqn
		Recall that the penalty $\gamma$ is given by \eqref{eq:penalty_ex_noncompact} and assume that $\gamma_{0,T}(\bar\eta)<\infty$ for $T>0$. Then, for all $0\le t\le T<\infty$,
			\bqn
				U(x,t)=\esssup_{\pi\in\mathcal{A}^x_t}\essinf_{\eta\in\q_{t,T}}
				\e^\eta\left[U(X_T^\pi,T)+\gamma_{t,T}(\meq^\eta)\bigg|\f_t\right],\label{mainex_optpb}					\eqn
		and the optimum is attained for the saddle point $(\bar\eta,\bar\pi)$ as given in \eqref{etastar}.
	\end{prop}
			
	The above result implies that the utility random field $U(x,t)$, given in \eqref{mainex}, and the penalty function $\gamma_{t,T}$ in \eqref{eq:penalty_ex_noncompact} constitute a robust forward criterion. For comparison, recall (cf. \cite{zari10}) that the random field
				\bqn
					U(x,t)=\ln x-\frac{1}{2}\int_0^t\hat\lambda^2_sds,\quad t\ge 0, x\in\R_+,\label{portus}
				\eqn
	constitutes a standard (non-volatile) forward criterion in the reference market $\hat\mep$ with market price of risk $\hat\lambda_t$, $t\ge 0$. We will see below that the above dynamics may be deduced by analysing the dual field, see \eqref{logeq} or \eqref{hjb}. However, the proof below is carried out directly in the primal domain.
	
	  \begin{proof}
	  	Fix $0\leq t\leq T<\infty$. To alleviate the notation, let $L_t=\int_0^t \hat\lambda_u d\hat W_u$. We have, with $1/p+1/q=1$ and $1/\tilde p+1/\tilde q=1$,
		\begin{equation*}
			\begin{split}
			\e^{\meq^{\bar \eta}}\Big[\int_0^T\hat\lambda_s^2ds\Big] & = \hat\e\left[D^{\bar\eta}_T\int_0^T\hat\lambda_s^2ds\right] = \hat\e\left[D^{\bar\eta}_T \langle L\rangle_T ds\right]\leq \left(\hat\e[(D_T^{\bar \eta})^p]\right)^{1/p}\left(\hat\e[\langle L\rangle_T^q]\right)^{1/q}\\
			&\leq \left(\hat\e\left[\mathrm{e}^{p\tilde pL_T-\frac{p^2\tilde p^2}{2}\langle L\rangle_T}\right]\right)^{\frac{1}{p\tilde p}} \left(\hat\e\left[ \mathrm{e}^{\kappa \langle L\rangle_T}\right]\right)^{\frac{1}{p\tilde q}} \left(\hat\e[\langle L\rangle_T^q]\right)^{1/q}<\infty,
			\end{split}
		\end{equation*}
		where we took $\tilde p<2\kappa$ and $p>1$ such that $\tilde q\left(\frac{p^2\tilde p}{2}-\frac{p}{2}\right)=\frac{p\tilde p(p\tilde p-1)}{2(\tilde p -1)}=\kappa$. It follows that $\gamma_{t,T}(\meq^{\bar\eta})<\infty$.
		Let
			\bq
				N^{\pi,\eta}_u:=U(X^\pi_u,u)+\int_t^u\frac{\delta_s}{2}|\eta_s|^2ds, \quad u\ge t.
			\eq
		Then, it suffices to show that $\e^{\bar\eta}\big[N^{\pi,\bar\eta}_T|\f_t\big]\le N^{\pi,\bar\eta}_t$, for all $\pi\in\mathcal{A}^x_t$, and that $\e^{\eta}\big[N^{\bar\pi,\eta}_T|\f_t\big]\ge N^{\bar\pi,\eta}_t$, for all $\eta\in\q_{t,T}$. For simplicity, and w.l.o.g., we show the claim in the case $t=0$. For $\pi\in\mathcal{A}^x$, the wealth process satisfies
			\bq
				dX^\pi_t=\pi_tX_t^\pi\sigma_t\left[\left(\hat\lambda_t+\eta^1_t\right)dt+dW^\eta_t\right],
				\quad t\le T,\quad X^\pi_0=x,
			\eq
		where $W^\eta_t$ is a Brownian motion under $\meq^\eta$. Due to the form of $U(x,t)$ and $\bar\pi$, a straight-forward application of It\^o's Lemma yields
		\bqq
			dN^{\bar\pi,\eta}_t
			&=&\frac{\delta_t}{1+\delta_t}\hat\lambda_t\left[\left(\hat\lambda_t+\eta^1_t\right)dt+dW^{\eta}_t\right]
				-\frac{1}{2}\left(\frac{\delta_t}{1+\delta_t}\hat\lambda_t\right)^2dt\\
				&& -\frac{1}{2}\frac{\delta_t}{1+\delta_t}\hat\lambda_t^2dt
				+\frac{\delta_t}{2}\left[\left(\eta^1_t\right)^2+\left(\eta^2_t\right)^2\right]dt\\
			&=& \frac{\delta_t}{1+\delta_t}\hat\lambda_t\eta^1_t dt
				+\frac{1}{2}\frac{\delta_t}{(1+\delta_t)^2}\hat\lambda_t^2dt
				+\frac{\delta_t}{1+\delta_t}\hat\lambda_tdW^{\eta}_t
				+\frac{\delta_t}{2}\left[\left(\eta^1_t\right)^2+\left(\eta^2_t\right)^2\right]dt\\
			&=&\frac{\delta_t}{2}\left[\left(\frac{\hat\lambda_t
				+\left(1+\delta_t\right)\eta^1_t}{1+\delta_t}\right)^2
				+\left(\eta^2_t\right)^2\right]dt
				+\frac{\delta_t}{1+\delta_t}\hat\lambda_t dW^{\eta}_t.
		\eqq
	Note that the quantity $\delta_t/(1+\delta_t)\in (0,1)$, so by the definition of $\gamma_{t,T}$ in \eqref{eq:penalty_ex_noncompact}, the process $\int_0^t\frac{\delta_s}{1+\delta_s}\hat\lambda_s dW^{\eta}_s$ is a martingale under $\meq^\eta$. It follows that $N^{\bar\pi,\eta}_t$ is a submartingale for all $\eta\in \q_{0,T}$ and a martingale for $\bar\eta$ as specified in \eqref{etastar}. On the other hand, it holds that
		\bqq
			U(X^\pi_T,T)+\int_0^Tg_s(\bar\eta_s)ds
			&=&\ln X^\pi_T-\int_0^T\frac{1}{2}\frac{\delta_s}{1+\delta_s}\hat\lambda_s^2
			-\frac{1}{2}\frac{\delta_s}{(1+\delta_s)^2}\hat\lambda_s^2ds\\
			&=&\ln X^\pi_T-\frac{1}{2}\int_0^T\left[\frac{\delta_s}{1+\delta_s}\hat\lambda_s\right]^2ds
			\;\;=\;\;\ln X^\pi_T-\frac{1}{2}\int_0^T\left(\hat\lambda_s+\bar\eta^1_s\right)^2ds.
		\eqq
		Since $ \e^{\bar\eta}\big[\ln X_T^{\pi}\big]\leq  \e^{\bar\eta}\big[\ln X_T^{\bar\pi}\big]$ for any strategy $\pi\in\mathcal{A}^x$, we conclude that
		\bq
			\e^{\bar\eta}\big[N^{\pi,\bar\eta}_T\big]
			\le \e^{\bar\eta}\big[\ln X_T^{\bar\pi}\big]-
			\e^{\bar\eta}\left[\frac{1}{2}\int_0^T\left(\hat\lambda_s+\bar\eta^1_s\right)^2ds\right]
			=\ln x=N_0,
		\eq		
	where the equality follows by a direct computation (see, also, p. 721 in \cite{karatzas91}).
	\end{proof}

	The investor's optimal behaviour described in Proposition \ref{buk} corresponds to strategies used in practice by some of the large fund managers. Specifically, the strategy, characterised by the  optimal fraction of wealth to be invested in the risky asset in \eqref{etastar}, is a fractional Kelly strategy. The investor invests in the  growth optimal (Kelly) portfolio corresponding to her best estimate of the market price of risk $\hat\lambda$. However she is not fully invested but instead chooses a leverage proportional to her trust in the estimate $\hat\lambda$. If $\delta_t\nearrow\infty$ (infinite trust in the estimation), then $\bar\pi_t\nearrow\hat\lambda_t/\sigma_t$ which is the Kelly strategy associated with the most likely model $\hat \mep$. On the other hand, if $\delta_t\searrow 0$ (no trust in the estimation), then $\bar\pi_t\searrow 0$ and the optimal behaviour is to invest nothing.
	
	We stress that $\hat\lambda$ and $\delta$ are the investor's arbitrary inputs. They might be data driven and come from an elaborate dynamic estimation procedure, be expert driven or simply come from a black box. In particular, there is no assumption that $\hat\lambda$ is a good estimate of the true market price of risk $\lambda$. In fact the latter never appears in the problem. It is crucial that the investor's utility function \eqref{mainex} evolves in function of the \emph{investor's perception of market} leading to a time-consistent behaviour solving \eqref{mainex_optpb}. This seem to capture well the investment practice -- in reality an investor never knows the ``true" model. Instead, she is likely to build (and keep updating) her best estimate thereof and act on it. This, as shown in Proposition \ref{buk}, can still lead to time-consistent optimal investment strategy.		
		In practice, the leverage has often a risk interpretation, e.g.\ it is adjusted to achieve a targeted level of volatility for the fund. In our framework, it is interpreted in terms of confidence $\delta$ in the estimate $\hat \lambda$. In practice, the leverage is adjusted rarely in comparison to the dynamic updating of the estimate $\hat \lambda$. Similarly, in our framework, the trust in one's estimation methods is likely to be adjusted on a much slower scale than the changes to the estimate itself.

		We note that the structure of the optimal investment strategy relies on the logarithmic form of the utility field \eqref{mainex}. Hence, on a finite time interval, one may expect a similar type of behaviour to be optimal also for some classical ambiguity averse utility maximization problem with logarithmic utility. The robust forward criterion in Proposition \ref{buk} presents, however, in many aspects the \emph{simplest} way of quantifying preferences corresponding to the investment behaviour in \eqref{etastar}. For example, these preferences are non-volatile while the value field associated with a deterministic utility function at a fixed horizon $T$, would be volatile\footnote{For a comparison with the variational criterion featuring (deterministic) logarithmic utility at some fixed horizon $T>0$, we refer to \cite{hernandez} for a stochastic factor model and Theorem 4.5 in \cite{laeven12} for the non-Markovian case.}. As further discussed in Section \ref{logexex}, robust forward criteria provides an alternative tool for the study of the link between investment strategies and the dynamic behaviour of the associated preferences. Proposition \ref{buk} illustrates this by providing, for a very popular investment strategy, the specification of compatible preferences with a particularly simple dynamic structure. \\

		\begin{rem}\label{fanta}
			For $\delta_t\equiv\delta$, the penalty function defined in \eqref{eq:penalty_ex_noncompact} corresponds to the entropic penalty function $\gamma(\meq)=\delta H(\meq|\hat\mep)$. For each fixed horizon $T$, the investment problem can then be rewritten as (cf. Remark 4.1 in \cite{tre}),
			\bqq
				u(x,0)&=&\sup_\pi\inf_\meq\Big(\e^\meq\left[U(X^\pi_T,T)\right]+\delta H(\meq|\hat\mep)\Big)\\
				&=&\sup_\pi -\delta \ln \e^{\hat\mep}\left[e^{-\frac{1}{\delta}U(X^\pi_T,T)}\right].
			\eqq
			Consequently, the problem is equivalent to a standard utility maximization problem with respect to the modified utility function $\tilde U(x,T)=-e^{-\frac{1}{\delta}U(x,T)}$ in the market $\hat\mep$. Therefore, it is then more natural to consider utility from intertemporal consumption (cf. \cite{bordigoni07,chen02,faidi11,jeanblanc12,lazrak03,skiadas03}). Note, however, that $\delta_t$, $t\ge 0$, is non-constant in our setting and, thus, the situation is different. 
		\end{rem}	
			
		\subsection{Criteria leading to time-inconsistent optimal investment strategies}\label{sec:ex:non_dc}	
		
		We turn now to an example of robust forward criteria which lead to time inconsistent optimal investment strategies. This complements our discussion in Sections \ref{secpenalty} and \ref{sec:time_consistent}. Lack of time-consistency of optimal strategies will be inherited from lack of dynamic-consistency of penalty functions. Here, for illustrative purposes, we develop an example where \eqref{kol} is violated in a rather unrealistically simplistic way.
		
		We work in the setting of Section \ref{sec:log_ex}. We set $\hat\lambda\equiv 0$ and we fix a family of bounded random variables $(\lambda^{t,T})$ with $0\leq t\leq T$, $T\geq 0$, each $\lambda^{t,T}$ being $\f_t$--measurable. Then we put
				\bqn 
		\gamma_{t,T}(\meq^\eta):=\left\{
			\begin{array}{ll}
				-\frac{T-t}{2}(\lambda^{t,T})^2 & \textrm{ if } (\eta^1_u,\eta^2_u)=(\lambda^{t,T},0),\ t\leq u\leq T,\\
				+\infty & \textrm{ otherwise.}
			\end{array}
			\right.
		\eqn
Clearly this is a degenerate and artificial example. At any time $t$, looking to invest on $[t,T]$, the investor believes only one model is feasible and gives it a well chosen negative penalty. The choice of this model changes arbitrary with $t$ and $T$ and there is no consistency requirement. Consider the extreme situation when all $\lambda^{t,T}$ are constant and $T$ fixed. Then, at time zero, the investor picks possibly different models which she will chose to believe when making investment decisions at $t$ for horizon $[t,T]$. It it is not surprising that this may lead to time-inconsistent investment strategies. However the flexibility of fixing the penalty $\gamma_{t,T}$ means that the dynamic-consistency of value functions, \eqref{eq:classicalDC} on $[0,T]$ or \eqref{eq:dynamic_consistency} in general, may be preserved.

		We let $U(x,t):= \ln x$ and $\eta^{t,T}_u:=0$ for $u<t$ and $\eta^{t,T}_u:=(\lambda^{t,T},0)$ for $t\leq u\leq T$. Note that by definition $\q_{t,T}=\{\meq^{\eta^{t,T}}\}$ so, using the classical results on log utility maximisation, we have
		\begin{equation*}
\begin{split}
		u(\xi,t,T)=& \ln \xi +\frac{1}{2}\e^{\meq^{\eta^{t,T}}}\left[\int_0^t |\eta_u^{t,T}|^2 du\Big|\f_t\right]+\gamma_{t,T}(\meq^{\eta^{t,T}})\\
		= & \ln \xi +\frac{1}{2}(T-t)(\lambda^{t,T})^2 +\gamma_{t,T}(\meq^{\eta^{t,T}}) = \ln \xi = U(\xi,t), \quad t\leq T,
\end{split}
\end{equation*}
and we conclude that $(U,\gamma)$ is a robust forward criteria and the value function is dynamically-consistent. Meanwhile, the resulting optimal strategy, at time $t$ when investing for the horizon $[t,T]$ is $\bar \pi^{t,T}_u = \frac{\lambda^{t,T}}{\sigma_t}$, $t\leq u\leq T$. Even when considering classical (robust) portfolio optimisation on $[0,T]$ these may be time inconsistent in the sense that $\bar\pi^{t,T}_u\neq \bar\pi^{u,T}_u$ for $t\leq u\leq T$. In our context of forward criteria, when $T$ is not fixed, the ``optimal strategy'' may be further \emph{horizon-inconsistent} in the sense that we may have $\bar \pi^{t,T}_t\neq \bar\pi^{t,T_1}_t$ for $t\leq T<T_1$. Hence, the ``optimal strategy" is not really a well defined concept since it may depend not only on when we make the decision but also on which horizon we want to consider. This is due to fundamental inconsistencies in the beliefs about feasible market models and violation of \eqref{kol}. The latter is in fact a non-trivial requirement. For example, penalty functions associated to convex risk measures via \eqref{penaltynew} do not satisfy \eqref{kol} in general. Whether \eqref{kol} is justified economically and empirically is one of interesting open questions resulting from our work, see also Remark 3.5 in Schied \cite{schied}.
				
		\subsection{On some important classes of robust forward criteria}\label{logexex}

		We discuss now the structure of robust forward criteria.
		Within the setup of Section \ref{sec:log_ex}, we describe the issue of non-uniqueness of robust forward criteria for given initial preferences. Examples of choices of specific classes of criteria where the uniqueness may be recovered are provided. 
		Particular attention is paid to the class of so called \emph{non-volatile} criteria, to which the main example studied in Section \ref{sec:log_ex} belongs.
		Here, we present a formal discussion motivating the definition of this class and illustrating its main features.

		\subsubsection{The structure of robust forward criteria}
		
		In the model-specific (non--robust) case, the robust forward performances are not uniquely specified from the initial condition. This is due to the flexibility of the volatility structure.		
		Before turning to the robust case, we recall the features of this structure. It holds that a random field is a (standard) forward criterion if, for all times $t\ge0$, it satisfies the SPDE
		\bqn
			dU(x,t)= \frac{1}{2}\frac{\left|\lambda_tU_x(x,t)+\sigma_t\sigma_t^+a_x(x,t)\right|^2}{U_{xx}(x,t)}dt+a(x,t)dW_t, \label{spdefirstf}
		\eqn
		equipped with the initial condition $U(x,0)=u_0(x)$. Similarly, the value function corresponding to the classical utility maximization problem satisfies (under some regularity conditions) the Backward SPDE \eqref{spdefirstf} equipped with the terminal condition $U(x,T)=U(x)$. We refer, respectively, to \cite{zarispde} and \cite{mania08} for a detailed presentation of these equations. A solution to the BSPDE \eqref{spdefirstf} equipped with a terminal condition is a pair of parameter-dependent processes $U(x,t)$ and $a(x,t)$ which are simultaneously obtained when solving the equation. Under some regularity conditions, the solution is unique (cf. \cite{mania08}). However, the presence of the volatility $a(x,t)$ implies that there might exist multiple stochastic terminal conditions, for all of which the associated solution satisfies $U(x,0)=u_0(x)$. 
		Put differently, starting from $u_0(x)=U(x,0)$ and solving forward in time we might arrive at different $U(x,T)$ depending on the choice of $a(x,t)$. It follows that the forward SPDE \eqref{spdefirstf} might have multiple solutions which are catalogued by their volatility $a(x,t)$. We refer to Section 1 in \cite{nad} for further discussion and axiomatic motivation. Likewise, even with a fixed penalty function, in order to specify robust forward criteria uniquely, we expect the need to impose further constraints. These could be either on the form of the primal/dual field or on the choice of volatility structure. We discuss both below.

	\subsubsection{Imposing constraints on the dual field}\label{sec:new_formal}

		We start with a formal discussion of a logarithmic example. Namely, we assume that $V(y,t)$ admits the representation
			\bqn
				V(y,t)=-\ln y+\int_0^tb_sds+\int_0^ta_s\cdot d\hat W_s,\label{log}
			\eqn
		for some processes $b_t$ and $a_t$ which are independent of $y$. Further, we assign to the measure $\meq^\eta$ (cf. \eqref{d}) a penalty given by\footnote{We recall that according to \cite{delbaen10}, it holds within a Brownian filtration that a dynamic penalty function is time-consistent (cf. \eqref{kol}) if and only if it is representable as in \eqref{penalty} for some $g_t(\cdot)$.}
		\bqn
			\gamma_{t,T}(\meq):=\e^\meq\bigg[\int_t^Tg_u(\eta_u)du\bigg|\f_t\bigg], \label{penalty}
		\eqn
		for some function $g:[0,\infty)\times\R^2\to[0,\infty)$, such that $g_t(\cdot)$ is convex, lower semicontinuous and satisfies the so called coercivity condition that $g_t(\eta)\ge-a+b|\eta|^2$ for some constants $a$ and $b$ (cf. (8.6) in \cite{tre}). For example, the choice of $g_t(\eta)=|\eta|^2+\infty\mathbf{1}_{\{|\eta|>\overline g\}}$ for some constant $\overline g>0$, ensures that $\gamma_{t,T}$ satisfies both Assumptions \ref{nonsing2} and \ref{gammatc}.\footnote{This follows e.g.\ from Lemma 3.1 in \cite{hernandez2} and the fact that $\q_{t,T}$ is weakly compact if and only if it is closed in $L^0$, see also discussion below Assumption \ref{nonsing2} above.} We let $\q=\cap_{T>0}\q_{0,T}$.

		Let $\mathcal{P}$ denote the set of all $\mathbb{F}$-progressively measurable processes $(\nu_t)_{t\ge 0}$ such that $\int_0^T\nu^2_tdt<\infty$ a.s.\ for all $T>0$. We assume that $(\hat\lambda_t)$ is in $\mathcal{P}$. For $\nu\in\mathcal{P}$, let 		
		\bqn
			Z^\nu_t:=\mathcal{E}\left(-\int\hat\lambda_sd\hat W^1_s-\int\nu_sd\hat W^2_s\right)_t.\label{z}
		\eqn 
		We note that $\z^e=\{Z^\nu: \textrm{$\nu\in\mathcal{P}$ and $Z^\nu_t$ is a $\hat\mep$-martingale on $[0,\infty)$}\}$ and write $\nu\in\z^e$ for $Z^\nu\in\z^e$. In particular, the assumption of NFLVR on finite horizons implies that $\nu_t\equiv 0\in\z^e$. 		
		According to Lemma \ref{charac}, in order for $V(y,t)$ and $\gamma_{t,T}$ to be self-generating, it then suffices\footnote{The stronger assumptions on $\gamma_{t,T}$ in Lemma \ref{charac} are used only to argue the necessity.} that for all $\nu\in\z^e$ and $\eta\in\mathcal{Q}$, the process
		\bqn
			M^{\eta\nu}_t:=V\Big(yZ_t^\nu/D_t^\eta,t\Big)+\int_0^tg(\eta_s)ds\label{m}
		\eqn
	is a $\meq^\eta$-sub-martingale, and there exist $\nu^*$ and $\eta^*$ for which it is a martingale. We recall that $\meq^\eta$ is given by $\frac{\mathrm{d}\meq^\eta}{\mathrm{d}\mep}|_{\f_t}=D^\eta_t$, with $D^\eta_t$ specified in \eqref{d}. 
		A straight-forward application of It\^o-Ventzell's formula, using that $V_{yy}(y,t)=1/y^2$, and formal minimization over $\nu_t$, yields that in order for $M^{\eta\nu}_t$ to satisfy this, the following relation must hold between $a_t$ and $b_t$: 		
			\bqn
				b_t=-\inf_\eta\bigg\{g(\eta)+\frac{\left(\eta^1+\hat\lambda_t\right)^2}{2}
				+a_t\cdot\eta\bigg\}, \quad a.s., \ t\geq 0. \label{logeq}
			\eqn
	We see that a given initial condition, a fixed penalty function $g(\cdot)$ and a volatility structure $a_t$ typically lead to a unique robust forward criteria: the drift is then specified via \eqref{logeq}. In consequence, for a given initial condition and specific penalty structure, a unique criterion may only be pinned down based on further specification of the dynamic properties of $U(x,t)$. 
				
		To conclude let us comment on another type of restriction on $V$. In many situations we might only be interested in solutions which are Markovian. For example, within a (Markovian) stochastic factor model, we could require that the utility field is a deterministic function of the underlying factors. This function must then solve a specific equation, closely related to the HJB equation associated with the classical value function within the same factor model. However, in the forward setting, the equation has to be solved forwards in time and is thus ill-posed. We refer to \cite{nad} for a study of such criteria in a model-specific setup. 
		
	\subsubsection{Imposing constraints on the volatility structure}
	
		We consider now constraints expressed in terms of the volatility structure.
		More specifically, we consider the class of criteria for which the volatility of the dual field (cf. \eqref{spdefirstf} and \eqref{log}) is identically zero; we refer to this class as \emph{non--volatile criteria}. Specifically, we assume that
		\bqn
			dV(y,t)=V_t(y,t)dt.\label{non-vol-def}	
		\eqn
	For standard forward criteria, this additional assumption specifies an interesting class of preferences; see \cite{mikepure,zari10}. In particular, we refer to \cite{mikepure} for a detailed discussion of the assumption \eqref{non-vol-def}. 
	Similarly to the example in Section \ref{sec:new_formal}, a straight-forward application of It\^o-Ventzell's formula and formal minimization over $\nu_t$, yields that in order for $M^{\eta\nu}_t$ (cf. \eqref{m}) to be a sub-martingale for each choice of $\nu$ and $\eta$ and a martingale at optimum, the random convex function $V(y,t)$ must solve the equation
			\bqn
				V_t(y,t)+\inf_{\eta}\Big\{
				g(\eta)+\frac{y^2V_{yy}(y,t)}{2}				\left(\eta^1_t+\hat\lambda_t\right)^2\Big\}=0,\quad a.s., t\ge 0.\label{hjb}
			\eqn
		
		This is a random equation, satisfied pathwise by the parameter-dependent process $V(y,t)$. The simplification from SPDE to a random PDE results from the restriction to non--volatile criteria. In particular, and in contrast to the SPDE case discussed above, we would expect that under suitable regularity assumptions \eqref{hjb} admits a unique solutions.
			
	Equation \eqref{hjb} might be viewed as a (dual) Hamilton-Jacobi-Bellman equation. In particular, a verification theorem stating that every well-behaved (convex) solution to \eqref{hjb} constitutes a robust forward criterion might be proven. However, to prove existence or explicitly solve this equation is hard. 
		In order to illustrate this, consider the case of no model-uncertainty, which corresponds to $g(\eta)=\infty$, $\eta\neq 0$. Then, equation \eqref{hjb} reduces to
			\bqn
				V_t(y,t)+\frac{\hat\lambda_t^2}{2}y^2V_{yy}(y,t)=0\quad a.s.,\quad t\ge 0.\label{nonrub}
			\eqn
		This equation characterizes standard non-volatile criteria in a model with market price of risk $(\hat\lambda_t)$. Equation \eqref{nonrub}, see \cite{mikepure,zari10}, is closely related to the (ill-posed) backward heat equation whose solutions only exist for a specific class of initial conditions, as characterised by Widder's theorem. Equation \eqref{hjb} inherits difficulties related to the equation being ill--posed but in addition is highly non-linear. Further, we have to ensure that its solution is adapted. 
		
		Note that in our main example, studied in Section \ref{sec:log_ex} above, the criterion \eqref{mainex} is logarithmic as well as non-volatile, and the appropriate form of the drift-term could, formally, be obtained by substituting the dual Ansatz $V(y,t)=-\ln y+\int_0^tb_sds$ into either of equations \eqref{logeq} or \eqref{hjb}. This is a rare case of an interesting and explicit solution to these equations.	We leave the analysis of \eqref{hjb} as a challenging problem open for further research. 
				
	\subsection{Equivalent standard (non-robust) forward criteria}\label{sec:equiv}
			
		We conclude with some remarks on the existence of equivalent forward criteria within a non--robust setting. To this end, observe that the optimal strategy $\bar\pi$ in \eqref{etastar} can also be interpreted as the Kelly-strategy associated with an auxiliary market with market price of risk $\bar\lambda_t$ given by 		
		\bqn
		 	\bar\lambda_t:=\hat\lambda_t+\bar\eta^1_t
		 	=\frac{\delta_t}{1+\delta_t}\hat\lambda_t. \label{mpr2}
		\eqn
		That is, the market price of risk $\hat\lambda$ that the investor thinks most likely, adjusted by the investor's trust in that estimation. This is closely related to the fact that the existence of the saddle-point $(\bar\pi,\bar\eta)$ implies that the optimal investment associated with the robust criterion \eqref{eq:penalty_ex_noncompact} and \eqref{mainex} coincides with the optimal investment corresponding to the non-volatile standard forward criterion (cf. \eqref{portus}),
			\bqn
				U(x,t)=\ln x-\frac{1}{2}\int_0^t\bar \lambda^2_sds,\quad t\ge 0, x\in\R_+,\label{tocome}
			\eqn
		specified in the market with the market-price of risk $\bar\lambda$. This is clear from the proof of Proposition \ref{buk}. 
		Such an equivalence can be established in far more generality. Indeed, given the existence of a saddle-point, the robust forward criterion, consisting of the pair $U(x,t)$ in \eqref{mainex} and $\gamma_{t,T}$ in \eqref{penalty}, ranks investment strategies in the same way as does the standard forward criterion,
	\bqn
		\tilde U(x,t)\;:=\; U(x,t)+\int_0^tg_s(\bar\eta_s)ds,\label{neutral}
	\eqn
	considered in the auxiliary market with the market price of risk $(\bar\lambda_t)$ with $\bar\lambda_t=\hat\lambda_t+\bar\eta_t^1$. 		
	Further, a formal application of Bayes' rule implies that the optimal strategy associated with the criterion \eqref{neutral}, is also optimal for the following forward criterion specified in the \emph{reference} market:
			\bqn
				D^{\bar\eta}_t \tilde U(x,t)\;=\;
				D^{\bar\eta}_t\Big(U(x,t)+\int_0^tg_s(\bar\eta_s)ds\Big).\label{tank}
			\eqn 
	Note that if $U(x,t)$ is a non-volatile criterion, $D^{\bar\eta}_t \tilde U(x,t)$ is in general volatile (cf. Theorem 4 in \cite{trob} for examples). Indeed, the assumption of non-volatility is market specific. Here (cf. \eqref{non-vol-def}), the non-volatility requirement is placed on the random field associated with the \emph{robust} criterion. If a saddle-point exists, the criterion is therefore non-volatile in the market specified by the optimal measure while the corresponding criterion in the reference market is volatile.

	For the class of robust forward criteria for which the above formalism can be made rigorous, the following holds: if the robust forward criterion admits an optimal strategy, then that strategy is optimal also for a specific standard (non-robust) forward criterion viewed in the reference market. Naturally, the latter criterion is defined in terms of the optimal $\bar\eta_t$, which is part of the solution to the robust problem and not a priori known. Nevertheless, on a more abstract level, this implies that \emph{viewed as a class of preference criteria, forward criteria can be argued to be 'closed' under the introduction of a certain type of model uncertainty}. For a similar conclusion in terms of the use of different numeraires, see Theorem 2.5 in \cite{karoui13} or Section 5.1 in \cite{karoui12}. This should also be compared to \cite{skiadas03}, where it was shown that to invest with respect to a given stochastic differential utility combined with a certain model uncertainty is equivalent to considering a modified stochastic differential utility within the reference model (therein, entropic penalty functions were considered but for a Brownian filtration and under some additional boundedness assumptions, the results can be extended also to variational preferences). For stochastic differential utilities as well as for forward criteria, the underlying reason is that the notion is general enough to allow for stochastic preferences. In particular, use of deterministic utility functions under model uncertainty is (under various conditions) equivalent to the use of specific stochastic utility functions for a fixed model. 
		%A fact which might be used to motivate the use of stochastic utility functions (cf. the use of different numeraires).
		
	The above implies that the preferences corresponding to (most) robust forward criteria may be embedded within the (standard) class of forward criteria. Nevertheless, we believe the example studied in this section illustrates that the notion of robust forward criteria is of interest. The aim of these criteria and the associated specific modelling of model uncertainty, is to disentangle the impact of the preferences originating from risk and model-ambiguity, respectively. A related and more involved question is under what conditions a given (volatile non-robust) forward criterion can be written as a non-volatile robust forward criterion with respect to some non-trivial penalty function. This question is left for future research. We also remark that the analysis herein and, thus, the above discussion, is restricted to measures equivalent to $\mep$. Considering absolutely continuous measures introduces further complexity (cf. \cite{schied} for the static case) but should not alter the main conclusions. In contrast, considering a larger set of possibly mutually singular measures would require new insights, see \cite{denis2013,nutz13}.

\section{Proofs}\label{sec:proofs}

	\subsection{Proof of Theorem \ref{main} and Proposition \ref{prop:existence}}\label{mainproof3}
	
	As discussed in Section \ref{sec:equivalence}, the proof of Theorem \ref{main} makes us of arguments and results presented in \cite{schied} and \cite{gordan}, respectively. To this end, we follows the notation in \cite{gordan} closely. In Section \ref{subsec:aux}, we prove conjugacy relations and existence of an optimizer for an auxiliary $\f_0$-measurable problem which we introduce below. In Section \ref{proof:reduction}, Theorem \ref{main} and Proposition \ref{prop:existence} are proven by reducing the general problem to the auxiliary one. 

\subsubsection{Notation}

	We let $0\le t\le T<\infty$ with $t$ and $T$ arbitrary, and $\kappa$ a random variable in $L^\infty_+(\f_t)$. We will typically consider $\kappa=\ind_A$, $A\in\f_t$, and use it to localise arguments to a set. We will also use the notation $Z_{t,T}\in\z^a_T$ to denote an element of the set $\{Z_{t,T}|Z\in\z^a_T\}$, and $Z\in\q_{t,T}$ to denote an element of the set $\{Z^\meq|\meq\in\q_{t,T}\}$. Unless stated otherwise, all the $L^p$-spaces, $p\in[0,\infty]$, are defined with respect to $(\Omega,\f_T,\mep|_{\f_T})$.

	We also let $\ket:=\left\{\int_t^T\pi_sdS_s:\pi\in \mathcal{A}_{bd}\right\}$ and $\cet:=\left(\ket-L^0_+\right)\cap L^\infty$. Note that the optimization over $\ket$ in \eqref{u3}, might be replaced by optimization over $\cet$. Then, for $\meq\in\q_{t,T}$, we introduce the function	
		\bq
			u^\meq_\kappa(\xi)=\sup_{g\in\cet}\e\left[\kappa\zetq U(\xi+g,T)\right],\qquad \xi\in L^\infty(\f_t).
		\eq
		
	Next, let $\mathcal{D}_{t,T}:=\big\{\zeta^*\in(L^\infty)^*:\langle\zeta^*,\zeta\rangle\le 0\textrm{ for all }\zeta\in\cet\big\}$ and, for $\eta\in L^1_+(\f_t)$, let $\deta:=\big\{\zeta^*\in\mathcal{D}_{t,T}:\langle\zeta^*,\xi\rangle=\langle\eta,\xi\rangle, \textrm{ for all }\xi\in L^\infty(\f_t)\big\}$.
	According to Lemma A.4 in \cite{gordan}, we have that
		\bqn\label{gordanlemma}
			\zeta^*\in\mathcal{D}_{t,T}\cap L^1_+
			\quad\textrm{if and only if}\quad
			\zeta^*=\eta\zet,
		\eqn 
	for some $\eta\in L^1_+(\f_t)$ and $\zet\in\z^a_T$. Note that the proof of this result uses that the market satisfies NFLVR on finite horizons. In turn, define the function  $\mathbb{V}^\meq_\kappa:\mathcal{D}_{t,T}\to(-\infty,\infty]$ by
		\bqn
			\mathbb{V}^\meq_\kappa(\zeta^*)
			:=\left\{\begin{array}{lcl}\e\left[\kappa\zetq V\left(\zeta^*/\big(\kappa\zetq\big),T\right)\right],&& \zeta^*\in L^1_+\;\textrm{and}\;\{\zeta^*>0\}\subseteq\{\kappa>0\}, \phantom{_\big|}\\
			\infty, && \textrm{otherwise}; \end{array}\right.\label{vvb}
		\eqn
	and the function $v^\meq_\kappa:L^1(\f_t)\to(-\infty,\infty]$ by
			\bqn
				v^\meq_\kappa(\eta):=\left\{\begin{array}{lcl}\inf_{\zeta^*\in\deta}\mathbb{V}^\meq_\kappa(\zeta^*),&& \eta\in L^1_+(\f_t), \phantom{_\big|}\\ \infty, && \eta\in L^1(\f_t)\setminus L^1_+(\f_t). \end{array}\right. \label{vb}
			\eqn	 
	Finally, we define the auxiliary value functions $u_\kappa:L^\infty(\f_t)\to(-\infty,\infty]$ by
	\bq
		u_\kappa(\xi)=\sup_{g\in\cet}\inf_{\meq\in\q_{t,T}}\e\Big[\kappa\Big(\zetq U(\xi+g,T)+\gamma_{t,T}(\meq)\Big)\Big],
	\eq
and $v_\kappa:L^1(\f_t)\to(-\infty,\infty]$ by
	\bq
		v_\kappa(\eta)=\inf_{\meq\in\q_{t,T}}\Big(v_\kappa^\meq(\eta)+\e\left[\kappa\gamma_{t,T}(\meq)\right]\Big).
	\eq

	\subsubsection{Results for the auxiliary value functions $u_\kappa$ and $v_\kappa$}\label{subsec:aux}

		We establish results for the $\f_0$-measurable value functions $u_\kappa$ and $v_\kappa$. Theorem \ref{main} and Proposition \ref{prop:existence} are then proven by reducing the problem to this case by taking expectations (cf. Section \ref{proof:reduction}). First, we consider the existence of a dual optimizer. \\

	\begin{prop}\label{additionallemma}
	Let $\eta\in L^1_+(\f_t)$. Then, there exists $(\bar\zeta^*,\bar\meq)\in\D^{\eta}_{t,T}\times\q_{t,T}$ such that 
		\bq
			v_\kappa(\eta)
			=
			\mathbb{V}_\kappa^{\bar\meq}(\bar\zeta^*)+\e\big[\kappa\gamma_{t,T}(\bar\meq)\big].
		\eq
	Moreover, the function $v_\kappa(\eta)$ is convex and lower semicontinuous with respect to the weak topology. 
		\end{prop}

	\begin{proof}
	Since $\eta\in L^1_+(\f_t)$, we deduce that 
	\bqn
		v_\kappa^{\meq}(\eta)=\inf_{\zeta^*\in\D^\eta_{t,T}}\mathbb{V}_\kappa^{\meq}(\zeta^*),\qquad \meq\in\q_{t,T}.\label{kork:n}	
	\eqn
	Hence, let $(\zeta^*_n,\meq_n)\in\D^{\eta}_{t,T}\times\q_{t,T}$, a sequence such that 
	\bqn
		\mathbb{V}_\kappa^{\meq_n}(\zeta^*_n)+\e\big[\kappa\gamma_{t,T}(\meq_n)\big]\to v_\kappa(\eta).	\label{convergence:n}	
	\eqn  
	Since $\q_{t,T}$ is weakly compact, there is a (sub) sequence $(\zeta^*_n,\meq_n)$ (also denoted by $n$) such that $\meq^n$ converges a.s. to some $\bar\meq\in\q_{t,T}$. From the Banach-Alaouglu theorem, we have that $\D_{t,T}$ is weak$^*$-compact. Hence, using the same argument as in the proof of Corollary A2 in \cite{gordan}, we deduce that there is a further sub sequence (denoted by $n$) such that $\zeta^*_n$ converges to some $\bar\zeta^*\in\D^\eta_{t,T}$ in the weak$^*$-topology. 
	
	Given Assumption \ref{nonsing2}, Proposition A.3 in \cite{gordan} may be applied to the utility field $\tilde U(x,t)=Z^\meq_{t,T}U(x,T)$. It then follows that
		\bqn \label{eq:brilliant}
			\mathbb{V}^\meq_\kappa(\zeta^*)
			=
			\sup_{\zeta\in L^\infty}\left(\e\left[\kappa Z^\meq_{t,T}U(\zeta,T)\right]-\langle\zeta^*,\zeta\rangle\right). 
		\eqn	
	Note that for each $\zeta\in L^\infty$, $\{Z^\meq_{t,T}U^-(\zeta,T):\meq\in\q_{t,T}\}$ is uniformly integrable due to Assumption \ref{nonsing2}. Applying Fatou's Lemma to the positive part then yields that the first term in \eqref{eq:brilliant} is lower semicontinuous as a function of $\meq\in\q_{t,T}$ with respect to a.s. convergence. The second term is continuous in $\zeta^*$ with respect to weak$^*$-convergence. Since the supremum preserves lower semicontinuity, it follows that the mapping $(\zeta^*,\meq)\to\mathbb{V}^\meq_\kappa(\zeta^*)$ is jointly l.s.c. with respect to the product topology on $\D^\eta_{t,T}\times\q_{t,T}$. Recall that by Definition \ref{penaltydef}, the mapping $\meq\to\e\left[\kappa\gamma_{t,T}(\meq)\right]$ is l.s.c. with respect to a.s. convergence. Combined with \eqref{convergence:n} applied to the subsequence defined above, this yields the existence of a minimizer $(\bar\zeta^*,\bar\meq)$.

	The convexity of $v_\kappa(\eta)$ follows immediately from the joint convexity of the mapping $(\zeta^*,\meq)\to\mathbb{V}^\meq_\kappa(\zeta^*)+\e\left[\kappa\gamma_{t,T}(\meq)\right]$ (cf. \eqref{eq:brilliant}). To argue the lower semicontinuity, we work as follows. Let $\eta_\alpha\in L^1_+$ such that $\eta_\alpha\to\eta$ weakly and let $(\zeta^*_\alpha,\meq_\alpha)$ be such that $v_\kappa(\eta_\alpha)\ge \mathbb{V}^{\meq_\alpha}_\kappa(\zeta^*_\alpha)+\e\left[\kappa\gamma_{t,T}(\meq_\alpha)\right]$. Using similar arguments as above, one can show that there exists a subsequence $(\zeta^*_\alpha,\meq_\alpha)$ converging in the product topology. Using the joint lower semicontinuity of the mapping $(\zeta^*,\meq)\to\mathbb{V}^\meq_\kappa(\zeta^*)+\e\left[\kappa\gamma_{t,T}(\meq)\right]$ yields the lower semicontinuity of $v_\kappa(\eta)$.  
	\end{proof}

	Next, we establish the conjugacy relations for $u_\kappa$ and $v_\kappa$ (cf. Proposition \ref{a1} below). To this end, we first establish two auxiliary lemmas. The first one follows by applying Propositions A1 and A3 in \cite{gordan} to the auxiliary stochastic utility function $\tilde U(x,T)$ given in \eqref{auxutility}. \\
	
	\begin{lem}\label{gordan}
		Let $\meq\in\q_{t,T}$ and $\kappa\in L^\infty_+(\f_t)$ be such that $\kappa\zet^\meq U(x,T)\in L^1$, $x\in\R$, and assume that $\tilde U(x,T):=\ind_{\kappa=0}U(x,T)+\ind_{\kappa>0}\zet^\meq U(x,T)$, $x\in\R$, (cf. \eqref{auxutility}) satisfies the non-singularity condition 3.3 in \cite{gordan}. Then, for any $\xi\in L^\infty(\f_t)$,
		\bq
			u^\meq_\kappa(\xi)=\inf_{\eta\in L^1(\f_t)}\left(v^\meq_\kappa(\eta)+\langle\xi,\eta\rangle\right).
		\eq		
	\end{lem}

	\begin{proof}
		Because $\meq\sim\mep$, we have that $\zet^\meq>0$, $\mep$-a.s. Moreover, due to assumption, $\tilde U(x,T)$, $x\in\R$, is integrable and satisfies the non-singularity condition. In consequence, Propositions A1 and A3 in \cite{gordan} can be applied to the auxiliary random utility function\footnote{Note that although $\zet^\meq U(x,T)\in L^1$, it is not a priori clear whether $Z^\meq_{t,s}U(s,x)\in L^1(\f_s)$, for $t<s<T$. Hence, it is not clear whether the associated random field is actually a utility field in the sense of Definition \ref{field} (the field could easily be adjusted in order for the utility and path regularity conditions to hold). However, Proposition A1 in \cite{gordan} only makes use of the slice $U(x,T)$ and can therefore be applied under the given assumptions.} $\tilde U(x,T)$, $x\in\R$. Using that on the set $\{\kappa>0\}$, $\tilde U(x,T)=\zetq U(x,T)$ and $\tilde V(y,T)=\zetq V\big(y/\zetq,T\big)$, application of these results yields
		\bqn
			u^\meq_\kappa(\xi)=\inf_{\zeta^*\in\mathcal{D}_{t,T}}
			\left(\mathbb{V}^\meq_\kappa(\zeta^*)+\langle\zeta^*,\xi\rangle\right).\label{gordaneq}
		\eqn
	The set $\mathcal{D}_{t,T}$ in \eqref{gordaneq} can w.l.o.g. be replaced by $\mathcal{D}_{t,T}\cap L^1_+$ (cf. \eqref{vvb}). According to \eqref{gordanlemma}, for each $\zeta^*\in\mathcal{D}_{t,T}\cap L^1_+$, there exist $\eta\in L^1_+(\f_t)$ and $\zet\in\z^a_T$ such that
		\bq
			\langle\zeta^*,\xi\rangle
			\;=\;\e\left[\eta\zet\xi\right]
			\;=\;\e\left[\eta\xi\right]
			\;=\;\langle\eta,\xi\rangle,\qquad \textrm{for all }\xi\in L^\infty(\f_t).
		\eq 
		Hence, $\zeta^*\in\mathcal{D}_{t,T}^\eta$. Conversely, $\mathcal{D}_{t,T}^\eta\subseteq\mathcal{D}_{t,T}$, for $\eta\in L^1_+$. In consequence, it follows from \eqref{gordaneq} that
		\bqq
			u^\meq_\kappa(\xi)&=&\inf_{\eta\in L^1_+(\f_t)}\inf_{\zeta^*\in\mathcal{D}_{t,T}^\eta}
			\left(\mathbb{V}^\meq_\kappa(\zeta^*)+\langle \eta,\xi\rangle\right)\\
			&=&\inf_{\eta\in L^1_+(\f_t)}\left(v^\meq_\kappa(\eta)+\langle\xi,\eta\rangle\right).
		\eqq	
		Since $v^\meq_\kappa(\eta)=\infty$ for $\eta\in L^1(\f_t)\setminus L^1_+(\f_t)$, the infimum may be taken over $L^1(\f_t)$. We easily conclude. 
	\end{proof}

	The next result is the present setting's analogue of Lemma 4.6 in \cite{schied} and is proven by use of the same lopsided minimax theorem. Together with Lemma \ref{gordan}, it is the cornerstone of the proof of the duality relation in Proposition \ref{a1} below.\\

\begin{lem}\label{schied}
	Assume that $\q_{t,T}$ is weakly compact and that $U(x,t)$ and $\gamma$ satisfy Assumption \ref{nonsing2}. Then,								
		\bqn
			\sup_{g\in\cet}\inf_{Z\in\q_{t,T}}\e\Big[\kappa\Big(Z U\left(\xi+g,T\right)+\gamma_{t,T}(Z)\Big)\Big]=
			\inf_{Z\in\q_{t,T}}\sup_{g\in\cet}\e\Big[\kappa\Big(Z U\left(\xi+g,T\right)+\gamma_{t,T}(Z)\Big)\Big].\label{text}
		\eqn			
\end{lem}

\begin{proof}
		For given $\xi\in L^\infty(\f_t)$ and $g\in\cet$, there exists $a>0$ such that $\xi+g\ge -a$ a.s. Hence, $U(\xi+g,T)\ge U(-a,T)$. For a sequence $(Z_n)_{n\in\mathbb{N}}$, $Z_n\in\q_{t,T}$, such that $Z_n\to Z$ a.s., we then use Fatou's Lemma to obtain
			\bqn
				\liminf_{n\to\infty}\e\Big[\kappa Z_n\Big(U(\xi+g,T)+U^-(-a,T)\Big)\Big]
				\ge \e\Big[\kappa Z\Big(U(\xi+g,T)+U^-(-a,T)\Big)\Big].\label{above}
			\eqn
		Since $\{Z_nU^-(-a,T)\}$, $Z_n\in\q_{t,T}$, is uniformly integrable due to Assumption \ref{nonsing2}, it follows that
			\bq
				\lim_{n\to\infty}\e\big[\kappa Z_nU^-(-a,T)\big]=\e\big[\kappa ZU^-(-a,T)\big],
			\eq
		and, thus, \eqref{above} implies that the function $Z\to\e[\kappa ZU(\xi+g,T)]$ is lower semicontinuous with respect to a.s.-convergence on $\q_{t,T}$. As $\q_{t,T}$ is convex and weakly compact, it is uniformly integrable. Hence, the mapping $Z\to\e[\kappa ZU(\xi+g,T)]$ is lower semicontinuous also with respect to convergence in $L^1$. This, in turn, yields weak lower semicontinuity as the function is convex (affine).
		
		According to Definition \ref{penaltydef}, $Z\to\e[\kappa \gamma_{t,T}(Z)]$ is also convex and weakly lower semicontinuous on $\q_{t,T}$, which is convex and weakly compact due to Assumption \ref{nonsing2}. On the other hand, for each $Z\in\q_{t,T}$, $g\to\e[\kappa ZU(\xi+g)]$ is concave on the convex set $\cet$. Applying the lopsided minimax theorem (cf. Chapter 6 in \cite{aubin84}), we obtain the desired result. 	
	\end{proof}

	The next result establishes the conjugacy relations between $u_\kappa$ and $v_\kappa$. This is the key result upon which the proof on the conditional versions in Theorem \ref{main} relies. The proof uses arguments similar to the ones used in \cite{schied}. However, while the arguments in \cite{schied} rely of the duality results in \cite{kramkov}, we here make use of Lemma \ref{gordan}. \\

\begin{prop}\label{a1}
	Assume that $\q_{t,T}$ is weakly compact and that $U(x,t)$ and $\gamma$ satisfy Assumption \ref{nonsing2}. Then, for all $\xi\in L^\infty(\f_t)$ and $\eta\in L^1_+(\f_t)$, it holds that
		\bq
			u_\kappa(\xi)=\inf_{\eta\in L^1(\f_t)}\big(v_\kappa(\eta)+\langle\xi,\eta\rangle\big)
			\quad\textrm{and}\quad
			v_\kappa(\eta)=\sup_{\xi\in L^\infty(\f_t)}\big(u_\kappa(\xi)-\langle \xi,\eta\rangle\big).
		\eq
	\end{prop}
	
	\begin{proof}
		
		From Lemma \ref{schied} we obtain
		\bqqn
			u_\kappa(\xi)&=&
			\sup_{g\in\cet}\inf_{\meq\in\q_{t,T}}\e\left[\kappa\left(\zetq U\left(\xi+g,T\right)+\gamma_{t,T}(\meq)\right)\right]\nn\\
			&=&\inf_{\meq\in\q_{t,T}}\Big(\sup_{g\in\cet}\e\left[\kappa\zetq U\left(\xi+g,T\right)\right]+\e\big[\kappa\gamma_{t,T}(\meq)\big]\Big)\nn\\
			&=& \inf_{\meq\in\q_{t,T}}\Big(u^\meq_\kappa(\xi)+\e\big[\kappa\gamma_{t,T}(\meq)\big]\Big).\label{proofkk}
		\eqqn
		
		Note that if $U(x^0,T)\in L^1$ for some $x^0\in\R$, then $U(x,T)\in L^1$ for all $x\in\R$. Indeed, due to concavity, for $x<x_0<y$ with $\lambda x + (1-\lambda)y=x_0$, it holds that
			\bq	
				\lambda\e\left[U\left(x,T\right)\right]+(1-\lambda)\e\left[U\left(y,T\right)\right]
				\le \e\big[U\big(x^0,T\big)\big].
			\eq		
		Since $\mep\in\q_{t,T}$ and $U(x,T)\in L^1$ due to assumption, we can w.l.o.g. replace the set $\q_{t,T}$ in \eqref{proofkk} by
			\bqn
				\q_{t,T}^\kappa:=\big\{\meq\in\q_{t,T}: \kappa \zet^\meq U(x,T)\in L^1\big\}. \label{qkappa}
			\eqn 	
		Due to Assumption \ref{nonsing2}, we may then apply Lemma \ref{gordan} for each $\meq\in\q_{t,T}^\kappa$, to obtain
		\bqq
			u_\kappa(\xi)&=&\inf_{\meq\in\q^\kappa_{t,T}}\left(\inf_{\eta\in L^1(\f_t)}\Big(v^\meq_\kappa(\eta)+\langle\xi,\eta\rangle\Big)+\e\left[\kappa\gamma_{t,T}(\meq)\right]\right)\\
			&=&\inf_{\eta\in L^1(\f_t)}\left(\inf_{\meq\in\q^\kappa_{t,T}}\Big(v^\meq_\kappa(\eta)+\e\left[\kappa\gamma_{t,T}(\meq)\right]\Big)+\langle\xi,\eta\rangle\right)
			\;\;=\;\; \inf_{\eta\in L^1(\f_t)}\big(v_\kappa(\eta)+\langle\xi,\eta\rangle\big),
		\eqq
		where it remains to argue the last step. To this end, note that for each $\zeta^*\in\deta$, $\eta\in L^1(\f_t)$, it holds that
		\bqq
			\e\left[\kappa\zet^\meq V\left(\zeta^*/\kappa\zet^\meq,T\right)\right]+\e[\xi\eta]
			&\ge & 		
			\e\left[\kappa\zet^\meq\left(U\left(\xi+g,T\right)-\zeta^*(\xi+g)/\kappa\zet^\meq\right)\right]+\e[\xi\eta]\\
			&=&
			\e\left[\kappa\zet^\meq U\left(\xi+g,T\right)\right]-\e\left[\zeta^*(\xi+g)\right]+\e[\xi\eta]\\
			&\ge&
			\e\left[\kappa\zet^\meq U\left(\xi+g,T\right)\right].
		\eqq
		Hence, it follows that $\q_{t,T}^\kappa$ can be replaced by $\q_{t,T}$ without loss of generality. This completes the proof of the first conjugacy relation. To argue that $v_\kappa$ is the convex conjugate of $u_\kappa$ it, thus, suffices to argue that $v_\kappa$ is convex and weakly lower semicontinuous. This follows from Proposition \ref{additionallemma} and we conclude. 
	\end{proof}
	
	\subsubsection{Proof of Theorem \ref{main} and Proposition \ref{prop:existence}}\label{proof:reduction}
	
	We prove the main results in Section \ref{sec:equivalence}. To this end, we follow the same procedure as in \cite{gordan} and reduce, by taking expectations, the problem to one involving the $\f_0$-measurable value functions $u_\kappa$ and $v_\kappa$. The results then follow from Propositions \ref{prop:existence} and \ref{a1} above. 
	
	First, we establish the existence of a dual optimizer.

	\begin{proof}[Proof of Proposition \ref{prop:existence}]
		Let $\kappa:=(\max(1,v(\eta;t,T))^{-1}\in\kappa\in L^\infty(\f_t)$. Note that $\kappa$ takes values in $[0,1]$ and w.l.o.g., we may assume that $\{\kappa>0\}\neq\emptyset$. Let $(\bar\zeta^*,\bar\meq)$ be a minimzer of $v_k(\eta)$, whose existence is ensured by Proposition \ref{additionallemma}. W.l.o.g., let $\{\bar\zeta^*>0\}\subseteq\{\kappa>0\}$. Observe that $v_\kappa(\eta)<\infty$ due to the definition of $\kappa$. Therefore, $\mathbb{V}_\kappa^{\bar\meq}(\bar\zeta^*)<\infty$ and, in turn, \eqref{vvb} yields that $\bar\zeta^*\in L^1$. Hence, $\bar\zeta^*\in \D_{t,T}\cap L^1_+$ and, thus, according to \eqref{gordanlemma} there exists $\bar Z\in\mathcal{Z}^a_T$ such that $\bar\zeta^*=\eta \bar Z_{t,T}$. In order to show that $(\bar\meq,\bar Z_{t,T})$ attains the essential infimum in \eqref{lv}, we argue by contradiction. To this end, assume that there exist $\varepsilon>0$, $\meq'\in\q_{t,T}$, $Z'_{t,T}\in\mathcal{Z}^a_T$ and a set $B\in\f_t$ with $\mep(B)>0$, such that
			\bqn
				\e^{\meq'}\Big[V\left(\eta\zet'/\zet^{\meq'},T\right)\Big|\f_t\Big]
		+\gamma_{t,T}(\meq')+\varepsilon
			<
				\e^{\bar\meq}\Big[V\left(\eta\bar Z_{t,T}/\zet^{\bar\meq},T\right)\Big|\f_t\Big]
		+\gamma_{t,T}(\bar\meq)		
			\quad\textrm{on $B$}.\;\; \label{assumpcont}
			\eqn
			Note that $B\subseteq \{\kappa>0\}$. Moreover, w.l.o.g. (scaling if necessary), we may choose $B$ such that $B\subseteq\{\kappa=1\}$. Let the random variable $\tilde\zeta^*\in L^1$ be given by $\tilde\zeta^*:=\eta(\zet'\ind_{B}+\bar Z_{t,T}\ind_{B^c})$. It follows that $\tilde\zeta^*\in\mathcal{D}_{t,T}$ and, thus, $\tilde\zeta^*=\eta\tilde Z_{t,T}$ for some $\tilde Z_{t,T}\in\mathcal{Z}^a_T$. Taking expectations on both side of \eqref{assumpcont} then yields
				\bq
					\mathbb{V}_\kappa^{\tilde\meq}(\tilde\zeta^*)+\e\big[\kappa\gamma_{t,T}(\tilde\meq)\big]
					-\varepsilon \mep(B)
					\le					
					\mathbb{V}_\kappa^{\bar\meq}(\bar\zeta^*)+\e\big[\kappa\gamma_{t,T}(\bar\meq)\big],				
				\eq
			which contradicts the choice of $(\bar\zeta^*,\bar\meq)$ as the minimizer.
		\end{proof}
	
	Next, in order to reduce the conditional conjugacy relations to the $\f_0$-measurable case, we establish an auxiliary lemma. \\
		
	\begin{lem}\label{reduce}
		For fixed $g\in\cet$ and $\xi\in L^\infty(\f_T)$, it holds that 	
		\bqn
			\e\left[\essinf_{\meq\in\q_{t,T}}\kappa
			\Big(\e\left[\left.\zet^\meq U\left(\xi+g,T\right)\right|\f_t\right]+\gamma_{t,T}(\meq)\Big)\right]
			=\inf_{\meq\in\q_{t,T}}\e\Big[\kappa\Big(\zet^\meq U\left(\xi+g,T\right)+\gamma_{t,T}(\meq)\Big)\Big].\label{reduceform}
		\eqn
	\end{lem}
	
	\begin{proof}
		
	The inequality '$\le$' is trivial. To show the reverse inequality, let
		\bq
			J(\meq):=\kappa\e\left[\left.\zet^\meq U\left(\xi+g,T\right)\right|\f_t\right]+\kappa\gamma_{t,T}(\meq),\quad \meq\in\q_{t,T}.
		\eq	
		Note that $\mep\in\q_{t,T}$. Moreover, since $U(x,T)\in L^1$ due to assumption, it holds that
			\bq
				J(\mep)=\kappa\e\left[U\left(\xi+g,T\right)|\f_t\right]+\kappa\gamma_{t,T}(\mep)\in L^1(\f_t).
			\eq
		Hence, w.l.o.g. the set $\q_{t,T}$ in the left hand side of \eqref{reduceform} can be replaced by $\tilde\q_{t,T}:=\{\meq\in\q_{t,T}:J(\meq)\in L^1(\f_t)\}$.

		Next, we claim that the set $\big\{J(\meq)|\meq\in\tilde\q_{t,T}\big\}$ is directed downwards. Indeed, let $\meq_1$, $\meq_2\in\tilde\q_{t,T}$ and define $A:=\{J(\meq_1)\le J(\meq_2)\}\in\f_t$. Let $\bar\meq$ given by $\frac{\mathrm{d}\bar\meq}{\mathrm{d}\mep}=\ind_AZ^{\meq_1}_T+\ind_{A^c}Z^{\meq_2}_T$. According to Lemma 3.3 in \cite{irina}, $\bar\meq\in\q_{t,T}$ and, furthermore,
		\bqq
			J(\bar\meq)&=&
			\kappa\e\left[\left.\big(\ind_AZ^{\meq_1}_{t,T}+\ind_{A^c}Z^{\meq_2}_{t,T}\big)
			U\left(\xi+g,T\right)\right|\f_t\right]+\kappa\gamma_{t,T}(\bar\meq)\\
			&=&\kappa\ind_A\Big(\e^{\meq_1}\left[U\left(\xi+g,T\right)|\f_t\right]+\gamma_{t,T}(\meq_1)\Big)
			+\kappa\ind_{A^c}\Big(\e^{\meq_2}\left[U\left(\xi+g,T\right)|\f_t\right]+\gamma_{t,t}(\meq_2)\Big)\\
			&=&\ind_AJ(\meq_1)+\ind_{A^c}J(\meq_2)\;=\;\min\{J(\meq_1),J(\meq_2)\}.
		\eqq
		In particular, this implies that $\bar\meq\in\tilde\q_{t,T}$. Consequently, it also follows that $\big\{J(\meq)|\meq\in\tilde\q_{t,T}\big\}$ is closed under minimization and so directed downwards. Hence, due to Neveu \cite{neveu75}, there exists a sequence $\meq_n\in\tilde\q_{t,T}$ such that $J(\meq_n)$ is decreasing and
			\bq
				\essinf_{\meq\in\tilde\q_{t,T}}J(\meq)=\lim_{n\to\infty} J(\meq
			_n).
			\eq
		 Use of the monotone convergence theorem then yields that
			\bq
				\e\left[\essinf_{\meq\in\tilde\q_{t,T}}J(\meq)\right]
				=\e\left[\lim_{n\to\infty}\downarrow J(\meq_n)\right]
				=\lim_{n\to\infty}\e\left[J(\meq_n)\right]
				\ge\inf_{\meq\in\tilde\q_{t,T}}\e\left[J(\meq)\right].
			\eq
		Using the above and the fact that $\tilde\q_{t,T}\subseteq\q_{t,T}$, we obtain
			\bq
				\e\left[\essinf_{\meq\in\q_{t,T}}J(\meq)\right]
				\;=\;\e\left[\essinf_{\meq\in\tilde\q_{t,T}}J(\meq)\right]
				\;\ge\;\inf_{\meq\in\tilde\q_{t,T}}\e\big[J(\meq)\big]
				\;\ge\;\inf_{\meq\in\q_{t,T}}\e\big[J(\meq)\big],
			\eq
		and we easily conclude. 		
	\end{proof}

		We are now ready to prove Theorem \ref{main}. We argue by contradiction, assuming that the conditional conjugacy relations does not hold. Taking expectations and applying Lemma \ref{reduce}, it then follows that the $\f_0$-measurable conjugacy relations between $u_\kappa$ and $v_\kappa$ are violated. In consequence, we may apply Proposition \ref{a1} to obtain a contradiction and conclude.

	\begin{proof}[Proof of relation \eqref{a5} in Theorem \ref{main}]\label{mainproof3sub}	
		First, we show that the inequality '$\le$'  holds. To this end, note that for fixed $\bar g\in\cet$ and $\bar\meq\in\q_{t,T}$, it trivially holds that
			\bq
				\essinf_{\meq\in\q_{t,T}}\Big(\e^\meq\left[\left.U(\xi+\bar g,T)\right|\f_t\right]+\gamma_{t,T}(\meq)\Big)
				\le \esssup_{g\in\cet}\e^{\bar\meq}\left[\left.U(\xi+g,T)\right|\f_t\right]+\gamma_{t,T}(\bar\meq),
			\eq
	with $\xi\in L^\infty(\f_t)$ and $\eta\in L^1_+(\f_t)$. Thus, it is immediate that
		\bqn
			u(\xi;t,T)\le \essinf_{\meq\in\q_{t,T}}\left(\esssup_{g\in\cet}\e^\meq\left[\left.U(\xi+g,T)\right|\f_t\right]+\gamma_{t,T}(\meq)\right).\label{tex}			
		\eqn			
	Next, for any $\meq\in\mathcal{M}^a_T$, we have that $S$ is a local martingale and, thus, so is the process $\int_0^t\pi_udS_u$, for all $\pi\in\A_{bd}$. Recall that for all $\pi\in\A_{bd}$, there exists $a>0$ such that $\int_0^t\pi_udS_u>-a$, $t\le T$. It follows that $\e^\meq[g]\le 0$, for all $g\in\cet$. In turn, since $U(x,T)\le V(y,T)+xy$, for all $x\in\R$, $y\ge 0$, it follows that
		\bqq
			\e^\meq\big[\left.U(\xi+g,T)\right|\f_t\big]
			&\le&\e^\meq\left[\left.V\left(\eta\zet/\zetq,T\right)\right|\f_t\right]
				+\e\big[\left.(\xi+g)\eta\zet\right|\f_t\big]\\
			&\le&\e^\meq\left[\left.V\left(\eta\zet/\zetq,T\right)\right|\f_t\right]+\xi\eta
			\qquad\textrm{a.s.,}
			\label{texmex}
		\eqq
	for all $\meq\in\q_{t,T}$, $Z_{t,T}\in\z^a_T$, $\xi\in L^\infty(\f_t)$, $g\in\cet$ and $\eta\in L^1_+(\f_t)$. In combination with \eqref{tex}, this implies that
		\bqq
			u(\xi;t,T)&\le& \essinf_{\meq\in\q_{t,T}}\left(\essinf_{Z\in\z^a_T}\e^\meq\left[\left.V\left(\eta\zet/\zetq,T\right)\right|\f_t\right]+\xi\eta+\gamma_{t,T}(\meq)\right)\\
			&=&v(\eta;t,T)+\xi\eta,			
		\eqq
	for all $\eta\in L^1_+(\f_t)$. This completes the proof of the first inequality.

	To prove the reverse inequality, we argue by contradiction and assume that there exist $\xi\in L^\infty(\f_t)$, $\varepsilon>0$ and $A\in\f_t$ such that
		\bq
			\essinf_{\meq\in\q_{t,T}}\Big(\e^\meq\left[\left.U(\xi+g,T)\right|\f_t\right]+\gamma_{t,T}(\meq)\Big)+\varepsilon \ind_A
			\le \e^\meq\left[\left.V\left(\eta\zet/\zet^\meq,T\right)\right|\f_t\right]+\gamma_{t,T}(\meq)+\xi\eta,			
		\eq
for all $g\in\ket$, $\zet\in\z^a_T$, $\meq\in\q_{t,T}$ and $\eta\in L^1_+(\f_t)$. Observe that $u(\xi;t,T)<\infty$ a.s. on $A$ and, w.l.o.g., we may assume that there is $M<\infty$ such that $u(\xi;t,T)\le M$ a.s. on $A$. Since $\kappa=1/\kappa$ on $A$, it follows by multiplying the above inequality by $\kappa=\ind_A$, taking expectations on both sides and applying Lemma \ref{reduce}, that
		\bq
			\inf_{\meq\in\q_{t,T}}\e\Big[\kappa\Big(\zetq U(\xi+g,T)+\gamma_{t,T}(\meq)\Big)\Big]+\varepsilon P(A)
			\le \e\bigg[\kappa\zetq V\bigg(\frac{\eta}{\kappa}\frac{\zet}{\zet^\meq},T\bigg)\bigg]+\e\left[\kappa\gamma_{t,T}(\meq)\right]+\e\left[\kappa\xi\eta\right],			
		\eq
	where the expression in the first expectation on the right hand side is defined to be zero on $A^c$. 
	According to \eqref{gordanlemma}, we have that for every $\zeta^*\in \mathcal{D}_{t,T}^\eta\cap L^1_+$ with $\eta\in L^1_+(\f_t)$, there exists $\zet\in\z^a_T$ such that $\zeta^*=\eta\zet$. Using this and taking the supremum over $g\in\ket$, we deduce that
		\bqn
			u_\kappa(\xi)+\varepsilon P(A)
			\le \mathbb{V}^\meq_\kappa(\zeta^*)+\e\left[\kappa\gamma_{t,T}(\meq)\right]+\langle \xi,\eta\rangle,	\label{prot}		
		\eqn
for all $\eta\in L^1_+(\f_t)$ such that $\eta=\eta\ind_A$, $\meq\in\q_{t,T}$ and $\zeta^*\in \mathcal{D}^\eta_{t,T}\cap L^1_+$. In consequence, for any $\eta\in L^1_+(\f_t)$ and $\meq\in\q_{t,T}$, the above inequality holds for all $\zeta^*\in \mathcal{D}^\eta_{t,T}$. Indeed, if $\zeta^*\notin L^1_+$ or $\eta\neq\eta\ind_A$, then it holds that $\mathbb{V}^\meq_\kappa(\zeta^*)=\infty$ (cf. \eqref{vvb}). Hence,
		\bq
			u_\kappa(\xi)+\varepsilon P(A)
			\le v^\meq_\kappa(\eta)+\e\left[\kappa\gamma_{t,T}(\meq)\right]+\langle \xi,\eta\rangle,			
		\eq
	for all $\eta\in L^1_+(\f_t)$ and $\meq\in\q_{t,T}$ and. Thus, in turn, since $u_\kappa(\xi)\le M<\infty$ due to the above choice of $\kappa$, we obtain
			\bq
				u_\kappa(\xi)<u_\kappa(\xi)+\varepsilon P(A) \le \inf_{\eta\in L^1(\f_t)}\left(v_\kappa(\eta)+\langle\xi,\eta\rangle\right).
			\eq	
	According to Proposition \ref{a1} we have, thus, obtained a contradiction and we easily conclude.
	\end{proof}

	\begin{proof}[Proof of relation \eqref{a6} in Theorem \ref{main}]
		The assertion \eqref{a5} implies that for all $\eta\in L^1(\f_t)$ and $\xi\in L^\infty(\f_t)$, $v(\eta;t,T)\ge u(\xi;t,T)-\xi\eta$. Hence, the inequality "$\ge$" follows directly.
		
		To prove the reverse inequality, we argue by contradiction and assume that there exist $\eta\in L^1_+(\f_t)$, $\varepsilon>0$ and $A\in\f_t$ such that
		\bq
			\essinf_{\meq\in\q_{t,T}}\Big(\e^\meq\left[\left.U(\xi+g,T)\right|\f_t\right]+\gamma_{t,T}(\meq)\Big)-\xi\eta+\varepsilon \ind_A
			\le \e^\meq\left[\left.V\left(\eta\zet/\zetq,T\right)\right|\f_t\right]+\gamma_{t,T}(\meq),			
		\eq
for all $g\in\ket$, $\xi\in L^\infty(\f_t)$, $\zet\in\z^a_T$ and $\meq\in\q_{t,T}$. Since $\eta$ might be replaced by $\eta\ind_A$ without violating the above inequality, we assume w.l.o.g. that $\eta=0$ on $A^c$. Next, multiplying the above inequality by $\kappa=\ind_A$, taking the expectation and using Lemma \ref{reduce} yields
		\bq
			\inf_{\meq\in\q_{t,T}}\e\Big[\kappa\Big(\zetq U(\xi+g,T)+\gamma_{t,T}(\meq)\Big)\Big]
			-\e\left[\xi\eta\right]+\varepsilon P(A)
			\le \e\bigg[\kappa\zetq V\bigg(\frac{\eta}{\kappa}\frac{\zet}{\zetq},T\bigg)\bigg]+\e\left[\kappa\gamma_{t,T}(\meq)\right].			
		\eq
	According to \eqref{gordanlemma}, for every $\zeta^*\in \mathcal{D}_{t,T}^\eta\cap L^1_+$, there exists $\zet\in\z^a_T$ such that $\zeta^*=\eta\zet$. Hence, it then follows that
		\bq
			u_\kappa(\xi)-\langle \xi,\eta\rangle+\varepsilon P(A)
			\le \mathbb{V}^\meq_\kappa(\zeta^*)+\e\left[\gamma_{t,T}(\meq)\right],			
		\eq
for all $\xi\in L^\infty(\f_t)$, $\meq\in\q_{t,T}$ and $\zeta^*\in \mathcal{D}^\eta_{t,T}\cap L^1_+$. Since $\mathbb{V}_\kappa^\meq(\zeta^*)=\infty$, for any other $\zeta^*\in \mathcal{D}^\eta_{t,T}$, the above inequality holds for all $\xi\in L^\infty(\f_t)$, $\meq\in\q_{t,T}$ and $\zeta^*\in \mathcal{D}^\eta_{t,T}$. Therefore,
		\bq
			u_\kappa(\xi)-\langle \xi,\eta\rangle+\varepsilon P(A)
			\le v^\meq_\kappa(\eta)+\e\left[\gamma_{t,T}(\meq)\right],			
		\eq
	for all $\xi\in L^\infty(\f_t)$ and $\meq\in\q_{t,T}$ and, thus, in turn,
			\bq
				\sup_{\xi\in L^\infty}\big(u_\kappa(\xi)-\langle\xi,\eta\rangle\big)
				<\sup_{\xi\in L^\infty}\big(u_\kappa(\xi)-\langle\xi,\eta\rangle\big)+\varepsilon P(A)
				\le v_\kappa(\eta),
			\eq	
	where we used that $\sup_{\xi\in L^\infty}\big(u_\kappa(\xi)-\langle\xi,\eta\rangle\big)<\infty$,  due to the choice of $\kappa$. According to Proposition \ref{a1} we have, thus, obtained a contradiction and we easily conclude.
	\end{proof}

\subsection{Proof of Propositions \ref{dcp}, \ref{prop:time_consistency} and \ref{characp}}\label{sec:proof_TC}

	In order to prove the results in Section \ref{sec:time_consistent}, we first establish two Lemmata.\\

\begin{lem}\label{dc}
	Let $V$ be a dual random field and $\gamma_{t,T}$ an admissible family of penalty functions such that either Assumption \ref{gammatc} holds, or \eqref{kol} holds and $v^-(\zeta;t,T)\in L^1(\f_t;\meq)$ for all $\zeta\in L^0(\f_t)$ and $\meq\in\tilde\q_{0,T}$, $t\le T$. Then, the dual value field $v(\cdot;t,T)$ and $\gamma_{t,T}$ are self-generating on $[0,T]$.
\end{lem}

\begin{proof}
	Fix $0\le s<t<T<\infty$. For $\meq\in\q_{0,T}$, we use the convention $\gamma_{0,t}(\meq)=\gamma_{0,t}(\meq_{|\f_t})$. Let $Z\in\z^a_t$ and $\meq\in\q_{s,t}$. Using Proposition \ref{prop:existence}, we denote by $Z^*$ and $\meq^*$ the optimal elements in $\z^a_T$ and $\q_{t,T}$, respectively, for which $v\big(\eta Z_{s,t}/Z^\meq_{s,t};t,T\big)$ is attained. Then, it holds that 	
	{\setlength{\arraycolsep}{-0.7cm}
	\bqqn	\label{eq:dc_dual_optimal_extended}
		\e\left[\left.Z_{s,t}^\meq v\left(\eta\frac{Z_{s,t}}{Z^\meq_{s,t}};t,T\right)\right|\f_s\right]+\gamma_{s,t}(\meq)&\nn\\
		&&=\e\left[\left.Z_{s,t}^\meq\left( 		
\e\left[\left.Z_{t,T}^{\meq^*} V\left(\eta\frac{Z_{s,t}}{Z^\meq_{s,t}}\frac{Z^*_{t,T}}{Z^{\meq^*}_{t,T}},T\right)\right|\f_t\right]+\gamma_{t,T}\left(Z_{t,T}^{\meq^*}\right)
		\right)\right|\f_s\right]+\gamma_{s,t}(\meq)\nn\\
		&&=\e\left[\left.Z_{s,t}^\meq Z_{t,T}^{\meq^*}V\left(\eta\frac{Z_{s,t}Z^*_{t,T}}{Z^\meq_{s,t}Z^{\meq^*}_{t,T}},T\right)\right|\f_s\right]+\gamma_{s,T}\left(Z_{s,t}^\meq Z_{t,T}^{\meq^*}\right)
		%&&\ge \essinf_{Z\in\z^a_T}\essinf_{\meq\in\q_{s,T}} \left\{\e\left[\left.Z_{s,T}^{\meq}V\left(\eta\frac{Z_{s,T}}{Z^{\meq}_{s,T}},T\right)\right|\f_s\right]+\gamma_{s,T}\left(Z_{s,T}^{\meq}\right)\right\}
		\;\; \ge\;\;v(\eta;s,T),
	\eqqn}	
	where it was used that $Z_tZ^*_{t,T}\in \z^a_T$ and that $\bar\meq\in\q_{s,T}$, with $\frac{\mathrm{d}\bar\meq}{\mathrm{d}\mep|_{\f_T}}=Z^\meq_t Z^{\meq^*}_{t,T}$. While this follows immediately from the fact that $\q_{t,T}$ is stable under pasting under assumption a), it follows from assumption b) by the following argument: $v^-(\zeta;s,T)\in L^1(\f_{T};\bar\meq)$, $\zeta\in L^0(\f_{T})$, implies (using that $v(\eta;s,t)$ is finite) that $\e^\meq\big[\gamma_{t,T}\left(\meq^*\right)|\f_s\big]<\infty$ and, thus, $\bar\meq\in\q_{s,T}$.

	Next, let $Z\in\z^a_T$ and $\meq\in\q_{0,T}$ be the optimal objects for which the infimum in $v(\eta;s,T)$ is attained. Note that due to \eqref{kol}, the fact that $\meq\in\q_{0,T}$, yields $\meq\in\q_{t,T}$ and $\meq|_{\f_t}\in\q_{s,t}$. Hence, it follows that
	%	{\setlength{\arraycolsep}{-0.9cm}
		\bqqn	\label{eq:dc_dual_optimal_restricted}
			v(\eta;s,T)&=&
			\e\left[\left.Z_{s,T}^{\meq}V\left(\eta\frac{Z_{s,T}}{Z^{\meq}_{s,T}},T\right)\right|\f_s\right]+\gamma_{s,T}\left(Z_{s,T}^{\meq}\right)\nn\\
			&&=\e\left[\left.Z_{s,t}^\meq\left(		
\e\left[\left.Z_{t,T}^{\meq} V\left(\eta\frac{Z_{s,t}}{Z^\meq_{s,t}}\frac{Z_{t,T}}{Z^{\meq}_{t,T}},T\right)\right|\f_t\right]+\gamma_{t,T}\left(Z_{t,T}^{\meq}\right)
		\right)\right|\f_s\right]+\gamma_{s,t}(\meq)\nn\\
		&&\ge \e\left[\left.Z_{s,t}^\meq v\left(\eta\frac{Z_{s,t}}{Z^\meq_{s,t}};t,T\right)\right|\f_s\right]+\gamma_{s,t}(\meq)
		\;\; \ge \;\;v(\eta;s,T),		
		\eqqn%}
	where the last inequality is due to \eqref{eq:dc_dual_optimal_extended}. In consequence, equality must hold, which combined with \eqref{eq:dc_dual_optimal_extended} yields
	\bq
		v(\eta;s,T)=
		\essinf_{\meq\in\q_{s,t}}
		\essinf_{Z\in\z^a_t}
		\left\{\e^\meq\left[\left.v\bigg(\eta\frac{Z_{s,t}}{Z^\meq_{s,t}};t,T\bigg)\right|\f_s\right]
		+\gamma_{s,t}(\meq)\right\}.
	\eq
	This completes the proof.
\end{proof}
	
	\smallskip
	
	\begin{lem}\label{charac}
	Let $V(y,t)$ be a random field associated with a utility random field (cf. \eqref{vd}), and let $\gamma_{t,T}$ a family of penalty functions satisfying \eqref{kol}. Then, the following two statements are equivalent:
	\begin{itemize}
		\item[i)]{$V(y,t)$ and $\gamma_{t,T}$ are self-generating.}
	
		\item[ii)]{For each $y>0$ and all $t\le T<\infty$,
			\bqn
				V(yZ_t/Z_t^\meq,t)\le \e^\meq\left[\left.V(yZ_T/Z_T^\meq,T)\right|\f_t\right]+\gamma_{t,T}(\meq),\label{villkor}
			\eqn
			for all $\meq\in\q_{t,T}$ and $Z\in\z^a_T$. Moreover, for each $\bar T>0$, there exists $\bar\meq\in\q_{0,\bar T}$ and $\bar Z\in\z^a_T$, such that \eqref{villkor} holds with equality for all $t\le T\le\bar T$.}
	\end{itemize}	
			
	Furthermore, if either a) the set $\q_{0,T}=\tilde\q_{0,T}$, $T>0$, or b) for any $T>0$ and all $\zeta\in L^0(\f_T)$, $V^-(\zeta,T)\in L^1(\f_T;\meq)$ for all $\meq\in\tilde\q_{0,T}$, then i) and ii) are equivalent to the following condition:
			\begin{itemize}
				\item[iii)]{For each $y>0$ and all $t\le T<\infty$, \eqref{villkor} holds for all $\meq\in\q_{t,T}$ and $Z\in\z^a_T$. Moreover, there is a $Z\in\z^a$ and a sequence $\big(\meq_{T^i}\big)$, $i\in\mathbb{N}$, with $\meq_{T^i}=\meq_{T^{i+1}}|_{\f_{T^i}}$ and $\meq_T:=\meq_{T^i}|_{\f_T}\in\q_{0,T}$, $T^i\ge T$, such that for all $0<t<T<\infty$, \eqref{villkor} holds with equality for $\meq_T$ and $Z_T$.}
			\end{itemize}
\end{lem}

	\begin{proof}
		First, we show that i) implies ii). To this end, assume that $V$ is self-generating, namely, for any $T>0$ and $t\le T$, it holds that
	\bq
		V(\eta,t)=
		\essinf_{\meq\in\q_{t,T}}\essinf_{Z\in\z^a_T}\Big\{
		\e^\meq\left[\left.V\left(\eta\zet/\zetq,T\right)\right|\f_t\right]
		+\gamma_{t,T}(\meq)\Big\},\quad \eta\in L^0_+.
	\eq
	Let $y>0$, $\tilde Z\in\z^a_T$ and $\tilde\meq\in\q_{t,T}$. Further, let $\eta:=y\tilde Z_t/Z_t^{\tilde\meq}$. Then, it follows that
		\bqq
			V(y\tilde Z_t/Z_t^{\tilde\meq},t)&=&
			\essinf_{\meq\in\q_{t,T}}\essinf_{Z\in\z^a_T}\bigg\{
			\e^\meq\bigg[V\bigg(y\frac{\tilde Z_t\zet}{Z_t^{\tilde\meq}\zetq},T\bigg)\bigg|\f_t\bigg]+\gamma_{t,T}(\meq)\bigg\}\\
			&\le&\e^{\tilde\meq}\big[V\left(y\tilde Z_T/Z_T^{\tilde\meq},T\right)\big|\f_t\big]+\gamma_{t,T}(\tilde\meq),
		\eqq
	which yields \eqref{villkor}.
	Next, let $\bar Z\in\z^a_{\bar T}$ and $\bar\meq\in\q_{0,\bar T}$ the optimal objects for which $v(y,0;\bar T)$ is attained; their existence is ensured by Proposition \ref{prop:existence}. 
	Let $Z^T:=\e[\bar Z|\f_T]$ and $\meq_T:=\bar\meq|\f_T$.
	Note that $\bar\meq\in\q_{0,\bar T}$, implies that $\bar\meq_T\in\q_{0,T}$ and $\bar\meq\in\q_{T,\bar T}$.
	Hence, by use of the same arguments as in \eqref{eq:dc_dual_optimal_restricted} (which makes use of \eqref{kol}) combined with the fact that $V(y,t)$ and $\gamma$ are self-generating, it follows that
	\bqqn	\label{eq:above}
		v(y;0,\bar T)
		&=&\e^{\bar\meq}\left[V\left(y \bar Z_{\bar T}/Z_{\bar T}^{\bar\meq},\bar T\right)\right]
			+\gamma_{0,\bar T}(\bar \meq)\nn\\
		%&=&\e^{\bar\meq}\left[\e^{\bar\meq}\left[V\left(y \bar Z_T\bar Z_{T,\bar T}/Z^{\bar\meq}_TZ_{T,\bar T}^{\bar \meq},\bar T\right)\big|\f_{T}\right]+\gamma_{T,\bar T}(\bar \meq)\right]
		%	+\gamma_{0,T}(\bar \meq)\\
		&\ge& \e^{\bar\meq}\left[v\left(y \bar Z_{T}/Z_{T}^{\bar\meq};T,\bar T\right)\right]
			+\gamma_{0,T}(\bar\meq)
			~\ge~ v(y;0,T).
	\eqqn
	By once again using the property of self-generation, it follows that \eqref{eq:above} must hold with equality. In consequence, $v(y;0,T)$ is attained for $Z^T$ and $\meq_T$, $T\le \bar T$. 
	We now argue that for $t\le T\le \bar T$, \eqref{villkor} holds as equality for $\bar Z$ and $\bar \meq$. To this end, assume contrary to the claim that there is $\varepsilon>0$ and $A\in\f_t$, $\mep(A)>0$, such that
		\bq
			V(y\bar Z_t/Z_t^{\bar\meq},t)+\epsilon\ind_A \le \e^{\bar\meq}\left[\left.V(y\bar Z_T/Z_T^{\bar\meq},T)\right|\f_t\right]+\gamma_{t,T}(\bar\meq).
		\eq
	Taking the expectation under $\bar\meq$ and using \eqref{kol} we, then, obtain
		\bqn \label{eq:middle}
			\e^{\bar\meq}\left[V(y\bar Z_t/Z_t^{\bar\meq},t)\right]
			+\gamma_{0,t}(\bar\meq)
			+\epsilon\bar\meq(A) 
			~\le~ \e^{\bar\meq}\left[V(y\bar Z_T/Z_T^{\bar\meq},T)\right]+\gamma_{0,T}(\bar\meq).
		\eqn
	However, due to the above, $v(y;0,t)$ is attained for $\bar Z_t=Z^t$ and $Z^{\bar\meq}_t=Z^{\meq_t}_t$. Hence, we obtain the contradiction $v(y;0,t)<v(y;0,T)$ which completes the proof of the claim.

	In order to prove that ii) implies i), it suffices to show that, for any $0<t<T<\infty$ and $\eta\in L^0_+(\f_t)$, it holds that
	\bqn
		V(\eta,t)\le
		\e^\meq\big[V\big(\eta\zet/\zetq,T\big)\big|\f_t\big]
		+\gamma_{t,T}(\meq),\label{13}
	\eqn
	for all $\meq\in\q_{t,T}$ and $Z\in\z^a_T$ and that there exists some $\hat\meq\in\q_{t,T}$ and $\hat Z\in\z^a_T$ for which equality holds. Note that \eqref{villkor} implies that for a simple, positive and $\f_t$-measurable random variable $\tilde\eta=\sum_{k=1}^ny_k\ind_{A_k}$, we have that
		\bqn
			V(\tilde\eta Z_t/Z_t^\meq,t)\le \e^\meq\left[\left.V(\tilde\eta Z_T/Z_T^\meq,T)\right|\f_t\right]+\gamma_{t,T}(\meq),\label{realp}
		\eqn
	for all $\meq\in\q_{t,T}$ and $Z\in\z^a_T$. Using similar arguments to the ones used in the proof of Theorem 3.14 in \cite{gordan}, this implies that \eqref{realp} holds for arbitrary $\tilde\eta\in L^0_+(\f_t)$. For any  $\meq\in\q_{t,T}$ and $Z\in\z^a_T$, \eqref{13} is then obtained by letting $\tilde\eta=\eta Z^\meq_t/Z_t$. Equality in \eqref{13} follows by a similar argument where all the inequalities become equalities by the choice of $\meq_T\in\q_{t,T}$ and $Z^T\in\z^a_T$ for which \eqref{villkor} holds with equality.

	Next, we show the equivalence between i) and iii). Given a sequence as specified in iii), part ii) holds trivially. Hence, it only remains to show that i) implies iii). To this end, let $T_1<T_2$. Further, let $Z^1\in\z^a_{T_1}$ and $\meq_1\in\q_{T_1}$ be the optimal arguments for which $v(y;0,T_1)$ is attained; their existence is ensured by Proposition \ref{prop:existence}. In turn, let $\meq^*\in\q_{T_2}$ and $Z^*\in\z_{T_2}$ be the optimal arguments for which $v\big(yZ^1_{T_1}/Z^{\meq_1}_{T_1};T_1,T_2\big)$ is attained, and define $\meq_2$ and $Z^2$ as follows:
				\bq
					\frac{\mathrm{d}\meq_2}{\mathrm{d}\mep|_{\f_{T_2}}}=Z^{\meq_1}_{T_1} Z^{\meq^*}_{T_1,T_2}\qquad\textrm{and}\qquad
					Z^2=Z^1_{T_1}Z^*_{T_1,T_2}.
				\eq
		By use of the same argument as in \eqref{eq:dc_dual_optimal_extended} (which makes use of \eqref{kol} and \eqref{stable}) combined with the fact that $V(y,t)$ and $\gamma$ are self-generating, it follows that $Z^2\in\z^a_{T_2}$, $\meq_2\in\q_{0,T_2}$ and that
			\bqqn
				v(y;0,T_1)
				%&=&
				%\e\left[Z^{\meq^1}_{T_1}V\left(yZ^1_{T_1}/Z^{\meq^1}_{T_1},T_1\right)
				%\right]+\gamma_{0,T_1}\left(\meq^1\right)\nn\\
				%&=&
				%\e\left[Z^{\meq^1}_{T_1}\left(\e\left[Z^{\meq^*}_{T_1,T_2}
%V\left(y\frac{Z^1_{T_1}Z^*_{T_1,T_2}}{Z^{\meq^1}_{T_1}Z^{\meq^*}_{T_1,T_2}},T_2\right)|\f_{T_1}\right]+\gamma_{T_1,T_2}\left(\meq^*\right)\right)\right]+\gamma_{0,T_1}\left(\meq^1\right)\nn\\
				&=&
				\e\left[Z^{\meq^1}_{T_1}Z^{\meq^*}_{T_1,T_2}
				V\left(y\frac{Z^1_{T_1}Z^*_{T_1,T_2}}{Z^{\meq^1}_{T_1}Z^{\meq^*}_{T_1,T_2}},T_2\right)\right]
				+\gamma_{0,T_2}\left(Z^{\meq_1}_{T_1} Z^{\meq^*}_{T_1,T_2}\right)\\
				&=&
				\e^{\meq^2}\left[V\left(yZ^2_{T_2}/Z^{\meq^2}_{T_2},T_2\right)\right]
				+\gamma_{0,T_2}\left(\meq^2\right)	
				~\ge~ v(y;0,T_2).	\nn \label{mohaha}
			\eqqn
		In consequence, equality must hold and, thus, $v(y;0,T_2)$ is attained for $Z^2$ and $\meq_2$.
		As argued above (cf. \eqref{eq:dc_dual_optimal_restricted}), it follows for any $T<T_2$, that $v(y;0,T)$ is attained for $Z=Z^2_T$ and $\meq=\meq_2|_{\f_T}$. 
		Subsequent repetition of the above pasting procedure then yields $Z\in\z^a$ and a sequence $\big(\meq_{T^i}\big)$, $i\in\mathbb{N}$, with $\meq_{T^i}=\meq_{T^{i+1}}|_{\f_{T^i}}$ and $\meq_T:=\meq_{T^i}|_{\f_T}\in\q_{0,T}$, $T^i\ge T$, such that for all $T>0$, $v(y;0,T)$ is attained for $Z_T$ and $\meq_T$. 
		In turn, by once again using arguments similar to the ones used to show that i) implies ii), we obtain that for any $t<T<\infty$, \eqref{villkor} holds as equality for $Z_T$ and $\meq_T$. Hence, iii) holds and we conclude. 	
	\end{proof}

	We now argue how the results in Section \ref{sec:time_consistent} follow from the above Lemmata. 
	First, Theorem \ref{main} and Lemma \ref{dc} readily yield Proposition \ref{dcp}. 
	Further, according to Proposition 3.9 in \cite{gordan}, the fact that $U(x,T)\in L^1(\f_T,\meq)$ for all $\meq\in\tilde\q_{0,T}$, $T>0$, implies that assumption b) of Lemma \ref{charac} holds. Hence, combined with Theorem \ref{main}, Lemma \ref{charac} yields Proposition \ref{characp}.

	Next, we argue Proposition \ref{prop:time_consistency}. W.l.o.g., let $t=0$. By use of the same arguments as in the proof of Lemma \ref{charac} (see i implies iii), it follows that there is $Z\in\z^a$ and a sequence $\big(\bar\meq_{T^i}\big)$, $i\in\mathbb{N}$, with $\bar\meq_{T^i}=\bar\meq_{T^{i+1}}|_{\f_{T^i}}$ and $\bar\meq_T:=\bar\meq_{T^i}|_{\f_T}\in\q_{0,T}$, $T\le T^i$, such that, for all $T\ge 0$, $v(y;0,T)$ is attained for $Z_T$ and $\bar\meq_T$.
	Due to the existence of a saddle-point and the duality between $u(\cdot;0,T)$ and $v(\cdot;0,T)$, it follows (cf. Theorem 2.6 in \cite{schied}) that for each $x\in\R$, there is $Z\in\z^a$ and a sequence $\big(\bar\meq_{T^i}\big)$ satisfying the above, such that
	\bq
		u(x;0,T)~=~
		\esssup_{\pi\in\A}%\essinf_{\meq\in\q_{t,T}}
		\e^{\bar\meq}\bigg[U\bigg(x+\int_0^T\pi_sdS_s,T\bigg)\bigg]
		+\gamma_{0,T}(\bar\meq).%, \quad x\in\R.%\label{u3}
	\eq
	By use of \eqref{kol}, the time--consistency now follows as for the classical utility maximization problem. For completeness, we argue this. To this end, let $T\le\bar T$. It follows that (cf. \eqref{eq:above}),
		\bqq
			u(x,0;\bar T)
	%		&=&
	%		\e^{\bar\meq}
	%		\Big[U\Big(x+\int_0^{\bar T}\bar\pi^{0,\bar T}_sdS_s,\bar T\Big)\Big]+\gamma_{0,\bar T}(\bar\meq)\\
			&=&
			\e^{\bar\meq}\Big[
			\e^{\bar\meq}\Big[U\Big(x+\int_0^{\bar T}\bar\pi^{0,\bar T}_sdS_s,\bar T\Big)|\f_T\Big]
			+\gamma_{T,\bar T}(\bar\meq)\Big]+\gamma_{0,T}(\bar\meq)\\
			&\le &
			\e^{\bar\meq}
			\Big[u\Big(x+\int_0^T\bar\pi^{0,\bar T}_sdS_s,T;\bar T\Big)\Big]+\gamma_{0,T}(\bar\meq)
			~\le~
			u(x,0;T). 
		\eqq
		In consequence, equality must hold and, thus, $\bar\pi^{0,T}_0=\bar\pi^{0,\bar T}_0$.  
		Next, let $u\le T$ and assume contrary to the claim that there is $\varepsilon>0$ and $A\in\f_u$ such that 
			\bqn	\label{eq:tc_primal}
				\e^{\bar\meq}\Big[U\Big(x+\int_0^T\bar\pi^{0,T}_sdS_s,T\Big)|\f_u\Big]+\gamma_{u,T}(\bar\meq)+\varepsilon\ind_A
			~\le~
			u\Big(x+\int_0^u\bar\pi^{0,T}_sdS_s,u;T\Big).
			\eqn
		Taking expectations under $\bar\meq$, using that $U$ and $\gamma$ are self-generating and that $\gamma$ satisfies \eqref{kol}, then yields (cf. \eqref{eq:middle}), 
			\bqn
				\e^{\bar\meq}
			\Big[U\Big(x+\int_0^T\bar\pi^{0,T}_sdS_s,T\Big)\Big]+\gamma_{0,T}(\bar\meq)
			~<~
			\e^{\bar\meq}
			\Big[U\Big(x+\int_0^u\bar\pi^{0,T}_sdS_s,u\Big)\Big]+\gamma_{0,u}(\bar\meq),
			\eqn
		which yields the contradiction $u(x,0;T)<u(x,0;u)$. Similarly, assuming the reverse strict inequality in \eqref{eq:tc_primal}, yields a contradiction. We easily conclude.

%Cite: \cite{hernandez2}; \cite{pelp}; \cite{schoen}. Zitkovic/Follmer in wrong alphabetical order/in \cite{pelp} i should be I.
% uncomment next line to change bibliography name to references
% \renewcommand{\bibname}{References}
%\bibliography{refs_2}        %use a bibtex bibliography file refs.bib
\bibliographystyle{plaintest}  %use the plain bibliography style
%\bibliographystyle{LastnameFirst}
%\bibliographystyle{FirstAuthorRev}

%\begin{thebibliography}{99}
		%\bibitem[]{}:, \emph{}, ().
%\end{thebibliography}		

\end{document}